\documentclass[aps,10pt,twocolumn,showpacs,pra,amsmath,amssymb,floatfix,superscriptaddress,dvipsnames]{revtex4-2}

\usepackage{amsthm}
\usepackage{mathtools}
\usepackage{physics}
\usepackage{bbm}
\usepackage{amsfonts}
\usepackage{amssymb}
\usepackage{graphicx}

\usepackage[T1]{fontenc}
\usepackage{bbold}
\usepackage{float}
\usepackage[utf8]{inputenc}
\usepackage{csquotes}
\usepackage{tikz-cd}
\usepackage[normalem]{ulem}
\usetikzlibrary{shapes, shapes.geometric, shapes.symbols, shapes.arrows, shapes.multipart, shapes.callouts, shapes.misc,decorations.pathmorphing}
\tikzset{snake it/.style={decorate, decoration=snake}}
\usepackage{tcolorbox}

\usepackage[colorlinks=true,linktocpage=true]{hyperref}
\hypersetup{
    allcolors  = {blue},
}

\theoremstyle{plain}
\newtheorem{theorem}{Theorem}

\newtheorem{definition}{Definition}
\newtheorem{proposition}{Proposition}
\newtheorem{corollary}{Corollary}

\definecolor{selectiveyellow}{rgb}{0.8, 0.0, 0.8}

\usepackage{array}
\usepackage{multirow}
\usepackage{xcolor}

\usepackage[utf8]{inputenc}
\usepackage[english]{babel}
\usepackage[T1]{fontenc}

\usepackage{amsfonts}
\usepackage{dsfont}

\usepackage{enumerate}
\usepackage{braket}
\usepackage{framed}
\usepackage{graphicx}

\usepackage{hyperref}
\usepackage{cleveref}
\usepackage{natbib}
\usepackage{csquotes}
\usepackage{xcolor}
\usepackage[normalem]{ulem}

\usepackage{bbold}

\newtheorem{lemma}[theorem]{Lemma}

\usepackage{soul}

\begin{document}

\newcommand{\inl}{INL -- International Iberian Nanotechnology Laboratory, Av. Mestre Jos\'{e} Veiga s/n, 4715-330 Braga, Portugal}
\newcommand{\inlshort}{INL -- International Iberian Nanotechnology Laboratory, Braga, Portugal}
\newcommand{\uff}{Instituto de F\'{i}sica, Universidade Federal Fluminense, Av. Gal. Milton Tavares de Souza s/n, Niter\'{o}i -- RJ, 24210-340, Brazil}
\newcommand{\uffshort}{Instituto de F\'{i}sica, Universidade Federal Fluminense, Niter\'{o}i (RJ), Brazil}
\newcommand{\cfum}{Centro de F\'{i}sica, Universidade do Minho, Campus de Gualtar, 4710-057 Braga, Portugal}
\newcommand{\cfumshort}{Centro de F\'{i}sica, Universidade do Minho, Braga, Portugal}
\newcommand{\cfup}{Departamento de F\'{i}sica e Astronomia, Faculdade de Ciências, Universidade do Porto, rua do Campo Alegre s/n, 4169–007 Porto, Portugal.}
\newcommand{\cfupshort}{Departamento de F\'{i}sica e Astronomia, Faculdade de Ciências, Universidade do Porto, Porto, Portugal.}
\newcommand{\ist}{Departamento de Engenharia Electrotécnica e de Computadores, Instituto Superior Técnico, Av. Rovisco Pais, 1049-001, Lisbon, Portugal.}
\newcommand{\istshort}{Departamento de Engenharia Electrotécnica e de Computadores, Instituto Superior Técnico, Lisbon, Portugal.}
\newcommand{\itel}{Instituto de Telecomunicações, Av. Rovisco Pais, 1049-001, Lisbon, Portugal.}
\newcommand{\itelshort}{Instituto de Telecomunicações,  Lisbon, Portugal.}

\newcommand{\infoumshort}{Departamento de Informática, Universidade do Minho, Braga, Portugal}
\newcommand{\inescshort}{HASLab, INESC TEC, Universidade do Minho, Braga, Portugal}

\title{Certifying nonstabilizerness in quantum  processors}

\author{Rafael Wagner}
\email[These authors contributed equally. \\ Corresponding author: ]{rafael.wagner@inl.int}
\affiliation{\inlshort}
\affiliation{\cfumshort}

\author{Filipa C. R. Peres}
\email[These authors contributed equally. \\ Corresponding author: ]{rafael.wagner@inl.int}
\affiliation{\inlshort}
\affiliation{\cfupshort}

\author{Emmanuel Zambrini Cruzeiro}
\affiliation{\istshort}
\affiliation{\itelshort}

\author{Ernesto F. Galv\~ao}
\affiliation{\inlshort}
\affiliation{\uffshort}

\date{\today}

\begin{abstract}
     Nonstabilizerness, also known as magic, is a crucial resource for quantum computation. The growth in complexity of quantum processing units (QPUs) demands robust and scalable techniques for characterizing this resource. We introduce the notion of \textit{set magic}: a set of states has this property if at least one state in the set is a non-stabilizer state. We show that certain two-state overlap inequalities, recently introduced as witnesses of basis-independent coherence, are also witnesses of multi-qubit set magic. We also show it is possible to certify the presence of magic across multiple QPUs without the need for entanglement between them and reducing the demands on each individual QPU. 
\end{abstract}

\maketitle

\textit{Introduction.}-- The certification of quantum devices is a crucial task~\cite{eisert2020quantum_certification}. One fundamental characteristic of quantum computing hardware is the ability to generate non-classical resources: quantum coherence~\cite{streltsov2017colloquium,baumgratz2014quantifying}, quantum entanglement~\cite{horodecki2009entanglement,erhard2020advances}, nonstabilizerness~\cite{bravyi2005universal}, Hilbert space dimension~\cite{brunner2008testing}, quantum contextuality~\cite{budroni2022kochen}  are all necessary resources for quantum information processing. The growth in complexity of near-term noisy devices~\cite{preskill2018NISQ,cai2023quantum,bharti2022NISQ} demands scalable and robust methods for witnessing non-classical properties.

In this Letter, we propose a technique for certifying the presence of nonstabilizerness in a network of quantum processing units (QPUs), without the need to entangle resources between separate units. Our protocol is efficient, robust, and based on estimating two-state overlaps $r_{i,j} = \text{Tr}(\rho_i \rho_j)$ (also known as fidelity, if one of the states is pure), as sketched in Fig.~\ref{fig: intro}. Interestingly, despite not entangling different units, we can use this scheme to certify nonstabilizerness among various QPUs. We envision many use cases. First, certifying the presence of nonstabilizer resources in a single QPU. Secondly, certifying the presence of nonstabilizerness generated by multiple QPUs in a network; for example, a referee collects overlap statistics from different QPUs and, after processing these values (i.e., calculating linear functionals of overlaps), can infer that at least some QPUs generated nonstabilizer resources. Finally, certifying nonstabilizerness present in $n$-qubit systems and not realizable by systems with fewer qubits.
\begin{figure}[t]
    \centering
    \includegraphics[width=0.38\textwidth]{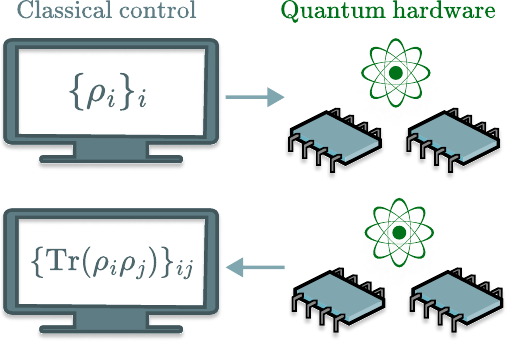}
    \caption{An experimentalist uses a quantum processing unit~(QPU) to evaluate two-state overlaps $r_{i,j} = \text{Tr}(\rho_i\rho_j)$ for a given set of states either generated by that same QPU or by a state-preparation device of a third party. With this information, and using the scheme we propose, they can benchmark nonclassical resources: nonstabilizerness, coherence, and Hilbert space dimension. This certification is agnostic to the procedure used by the QPU to compute the overlaps.}
    \label{fig: intro}
\end{figure}

Many tools for certifying nonstabilizerness exist. However, as the 
majority is aimed at the task of  \emph{resource quantification} -- in the formal resource theoretic sense~\cite{chitambar2019quantumRTs} -- rather than the simpler one of \emph{witnessing}, most protocols have strict requirements. Most commonly, there is a need for: (i) full tomographic information of quantum states~\cite{campbell2011catalysis}, (ii) purity of states (or specific subclasses of states)~\cite{haug2023scalable,gross2021schurweyl,haug2023efficient}, (iii) additional entanglement generation~\cite{tirrito2023quantifying,turkeshi2023measuring}, or (iv) vertex characterization of the stabilizer polytope~\cite{liu2022many,bravyi2019simulationofquantum,bravyi2016improved,bravyi2016trading,veitch2014resource}. We avoid all these requirements. Our certification scheme depends on two promises of the device-generated statistics: (a) the data is described by two-state overlaps and (b) the device consists of a multi-qubit system. For certain witnesses, it will be possible to relax the last requirement. 

Our method is based on showing that some inequalities, initially introduced as coherence witnesses~\cite{galvao2020quantum,wagner2023inequalities}, can also be used as witnesses of nonstabilizerness. The basis-independent nature of such inequalities implies the existence of a relational notion of nonstabilizerness defined for an ensemble of states, that we term \textit{set magic}. The idea of viewing ensemble-based resources has recently been widely studied within quantum resource theories~\cite{designolle2021set,salazar2022resource,buscemi2020complete,martins2020quantum,uola2019quantifying,ducuara2020multiobject}. This novel foundational understanding is at the core of the application we envision, allowing us to distribute a quantum certification protocol among various parties in a network of QPUs.

\textit{Background: Stabilizer subtheory.}-- In this work, we will focus on the stabilizer subtheory of $n$-qubit systems. We say that a pure quantum state $\vert \psi \rangle$ is a stabilizer state when it is the eigenvector with eigenvalue $+1$ for all elements of a maximal abelian subgroup of the Pauli group $\mathcal{P}_n$. The Pauli group is formed by all possible $n$-qubit Pauli operators, multiplied by phases $\pm 1$ and $\pm i$. The dynamics of this subtheory is described in terms of Clifford operations, defined to be those that preserve the Pauli group under conjugation, i.e., $C \mathcal{P}_n C^\dagger = \mathcal{P}_n$. 

For many different architectures, it is relatively easy to perform Clifford operations and prepare stabilizer states. This subtheory is crucial to applications in fault-tolerant magic-state injection schemes and quantum error correction, and much work has been done to understand its geometrical properties. Ref.~\cite{garcia2017onthegeometry} showed that stabilizer states have a fixed overlap structure. Recall that any two pure states $\vert \psi\rangle, \vert \phi \rangle \in \mathcal{H}$ can be orthogonal ($|\langle \psi | \phi \rangle| = 0$), parallel ($|\langle \psi | \phi \rangle| = 1$) or oblique ($0<|\langle \psi|\phi \rangle| < 1$). When two \emph{pure} stabilizer states are oblique, their two-state overlap obeys the following result.
\begin{theorem}[Adapted from Ref.~\cite{garcia2017onthegeometry}]\label{theorem: quantization}
    Given two oblique $n$-qubit pure stabilizer states $\vert \psi \rangle$ and $\vert \phi \rangle$ their overlap is given by $\vert \langle \phi \vert \psi \rangle \vert^2 = 2^{-k}$ for some $k \in \{1,2,\dots,n\}$.
\end{theorem}

The convex hull of all $n$-qubit pure stabilizer states, for each fixed value of $n\geq 1$, forms a polytope that we will refer to as STAB. Any state outside of STAB is termed in the literature \textit{magic} or \textit{nonstabilizer}. Note that Theorem~\ref{theorem: quantization} presents a simple method for witnessing magic of \emph{pure} states using overlaps. Any deviation from the values $1/2^k$ witnesses the presence of such a resource. Beyond that, purity allows efficient schemes to quantify nonstabilizerness. Refs.~\cite{haug2023efficient,haug2023scalable} provide \emph{quantifying} protocols that are close to being optimally efficient. However, without the assumption of purity (which is never perfectly attained experimentally), any overlap is possible by states inside of STAB, and the task of witnessing magic becomes non-trivial. The test we propose, although based on two-state overlaps, does not require purity of states.

\textit{Background: Coherence witnesses based on two-state overlaps.}-- Initially motivated by benchmarking various resources in linear optical devices -- Hilbert space dimension~\cite{giordani2021witnesses}, quantum coherence~\cite{giordani2021witnesses}, indistinguishability~\cite{brod2019witnessing,giordani2020experimental} -- Refs.~\cite{galvao2020quantum,wagner2023inequalities} proposed an 
inequality formalism based solely on linear functionals of two-state overlaps, defined by $r_{i,j} = \text{Tr}(\rho_i\rho_j)$ for any two states $\rho_i$ and $\rho_j$ over the same Hilbert space. 

Consider edge-weighted graphs $(\mathcal G,r)$ where $\mathcal G = (V(\mathcal G),E(\mathcal G))$ is a graph~\footnote{A graph is an ordered pair $(V(\mathcal G),E(\mathcal G))$ where $V(\mathcal G), E(\mathcal G)$ are sets. The first set $V(\mathcal G)$ is interpreted as a set of nodes (or vertices) for the graph, hence a simple set of labels. The second set $E(G)$ corresponds to elements of the form $e = \{i,j\}$ such that $i,j \in V(\mathcal G)$, i.e., the edges of $\mathcal G$. }, and $r: E(\mathcal G) \to [0,1]$ is a function. We refer to fully connected finite simple graphs $\mathcal{G}$ as \textit{event graphs}~\footnote{A finite graph is one in which $|V(\mathcal{G})|<\infty$. A simple graph is undirected, no pair of nodes can have more than one edge, and there exist no edges of the form $\{v,v\}$ for $v \in V(\mathcal{G})$. Finally, a fully connected graph is one in which for any two vertices $v,w$ there exists a family of edges connecting the two.}, as introduced in Ref.~\cite{wagner2023inequalities}. Any  $(r_e)_{e \in E(G)}$ is merely a tuple of numbers. We will be interested in the problem of deciding when the numbers $r_{i,j}$ in these tuples are realizable by quantum states, thus equating to $\text{Tr}(\rho_i \rho_j)$; this is an instance of a \emph{quantum realization problem}~\cite{fraser2023realization}. More formally, let $\mathcal{D}(\mathcal{H})$ represent the set of all quantum states with respect to a system $\mathcal{H}$ and consider a finite state ensemble $\underline\rho \equiv \{\rho_i\}_i \subset \mathcal{D}(\mathcal{H})$. Denote $r(\underline\rho) \equiv r(\{\rho_i\}_i) := (\text{Tr}(\rho_i\rho_j))_{\{i,j\} \in E(\mathcal G)}$; we will refer to $r(\underline\rho)$ as a \emph{quantum realization} for a given edge-weight $r$. Tuples $r \in [0,1]^{|E(\mathcal{G})|}$ can have any number of quantum realizations, including none at all~\cite{wagner2023inequalities}.

We proceed to discuss the notion of coherence for a state ensemble~\cite{designolle2021set,galvao2020quantum}, that differs from the commonly described basis-dependent view on coherence~\cite{baumgratz2014quantifying,streltsov2017colloquium} and its witnesses~\cite{napoli2016robustness,wu2021experimental,silva2023coherencewitnessing}. It also differs from the basis-independent coherence discussed in  Refs.~\cite{ma2019operational,radhakrishnan2019basis,yao2016frobenius} defined with respect to a single state $\rho$ instead of sets $\underline\rho$. When there exists some unitary $U$ such that a set of states $\underline\rho$  satisfy $U \underline\rho U^\dagger = \underline\sigma$ with $\underline\sigma$ some set of diagonal density matrices, we say that the state ensemble  $\underline\rho$ is set-incoherent~\footnote{Equivalently, a set of states is incoherent iff all its elements pairwise commute.}. Otherwise, we say that it is set-coherent~\cite{designolle2021set}. Set coherence is a basis-independent property of a set of states. Importantly, for any event graph $\mathcal{G}$, it is possible to bound two-state overlaps $r$ realized by incoherent ensembles: in such cases, incoherent states in the ensemble satisfy $\underline\sigma \ni \sigma = \sum_{\lambda \in \Lambda} \sigma_{\lambda \lambda} \vert \lambda \rangle \langle \lambda \vert$, for some $\mathcal{H}$ with $d=\dim(\mathcal{H})$ and some basis $\Lambda = \{\vert \lambda \rangle\}_{\lambda = 1}^d$.

For any fixed event graph $\mathcal{G}$, the set of all possible edge-weights $r$ realizable by some set-incoherent ensemble $\underline\sigma$ forms a full-dimensional convex polytope, denoted by $\mathfrak{C}(\mathcal{G})$~\cite{wagner2023inequalities}. Some convex polytopes $\mathfrak{C}(\mathcal{G})$ have been completely characterized for certain event graphs $\mathcal{G}$. The facet-defining inequalities for $\mathfrak{C}(\mathcal C_m)$, with $\mathcal{C}_m$ the cycle graph of $m$-nodes, were presented in Ref.~\cite{galvao2020quantum} and are the so-called $m$-cycle inequalities

\begin{equation}
    c_m(r) := -r_{e} +\sum_{{e'} \neq {e}} r_{e'} \leq m-2,\;\; \text{for each $e \in E(\mathcal C_m)$.} \label{eq:ncycle}
\end{equation}
For any convex-linear functional $h(r)$, we denote $h(r(\underline\rho))$ its value attained by some quantum realization $r(\underline\rho)$. Violations $c_m(r(\underline\rho))>m-2$ witness the impossibility of the overlaps to be realized by incoherent ensembles of states, i.e., they witness set coherence of any such ensemble $\underline\rho$. We resolve two open questions from Ref.~\cite{galvao2020quantum}. First, in Appendix~\ref{app: SDP}, using semidefinite programming (SDP) relaxation tools~\cite{tavakoli2023semidefinite}, we obtain tight maximal violations of such inequalities up to $m=7$, and lower bounds up to $m=20$. Secondly, in Appendix~\ref{app: proof of cycle results}, we provide a proof of the following theorem.
\begin{theorem}\label{theorem: qubits cycle attainable}
    The maximal quantum values of $c_m(r)$, for any $m\geq 3$ are attainable with sets $\underline\psi \subseteq \mathcal{D}(\mathbb{C}^2)$ of pure single-qubit states.
\end{theorem}
We also show that all facet-defining overlap inequalities for $\mathfrak{C}(\mathcal{G})$, for any event graph $\mathcal{G}$, are maximally violated by sets of pure states, meaning that, in particular, pure states maximally violate $c_m(r)\leq m-2$. 

Another relevant family of inequalities, defined recursively, is the following:
\begin{equation}\label{eq: family}
    h_{m}(r) = h_{m-1}(r)+r_{1m}-\sum_{i=2}^{m-1}r_{im} \leq 1\,.
\end{equation}
The sequence above starts with $h_{3}(r) = r_{12}+r_{13}-r_{23} \equiv c_3(r)$ and   defines novel inequalities for any integer $m>3$. Each $h_m$ inequality was shown in Ref.~\cite{wagner2023inequalities} to be facet-defining for the polytope $\mathfrak{C}(\mathcal{K}_m)$, where $\mathcal{K}_m$ is the complete graph with $m$ nodes. It has been shown for up to $12$-qubit systems that $h_m(r)\leq 1$ cannot be violated by sets $\underline\rho \subseteq \mathcal{D}(\mathbb{C}^d)$ where $d\leq m-2$~\cite{giordani2023experimental}. Numerically, we can see maximal quantum violations for $m$ states $\underline\rho$ with dimension $d\geq m-1$. 

\textit{Generalizing magic to state ensembles.}-- We start the presentation of our main results by introducing the notion of set magic. 
\begin{definition}[Set magic]
Let $\underline\rho \subseteq \mathcal{D}(\mathcal{H})$ be a finite set of states and $d=\dim(\mathcal{H})<\infty$. We say $\underline\rho$ is \textit{set-magic} when there exists no unitary $U: \mathcal{H} \to \mathcal{H}$ such that $U\underline\rho U^\dagger = \underline\rho^{\prime}$ is a set of states within the stabilizer polytope. 
\end{definition}
Set magic implies magic of \textit{some} element in the set. Notably, the converse is \emph{not} true, as we will see in a moment. The definition above is a straightforward translation of set coherence~\cite{designolle2021set} to the context of magic as a quantum resource. Later, we will show that, under certain conditions, it is possible to bound the number of states $\rho \in \underline\rho$ that are outside of STAB, instead of merely witnessing that some are.

Set magic and set coherence are different notions. A simple way to see this is that various sets of stabilizer states $\underline\rho_{\mathrm{STAB}}$ (e.g., $\{\vert 0\rangle, \vert +\rangle,\vert+_i\rangle\}$) are set coherent, without being set magic. We can present another (more subtle) difference between the two notions. Given any basis-\emph{dependent} coherent pure state $\vert \psi \rangle$, with respect to some basis of reference $\{\vert i\rangle\}_i$, we have that $\{\vert \psi \rangle, \vert i\rangle\}$ must be set-coherent. In general, given a magic state $\vert \phi \rangle$ and a stabilizer state $\vert s \rangle \in $ STAB, we have that $\{\vert \phi \rangle, \vert s \rangle \}$ will \emph{not} necessarily be set-magic. For example, the set $\{\vert 0\rangle, \vert T \rangle \}$, with $\vert T \rangle := \frac{1}{\sqrt{2}}(\vert 0\rangle + e^{i\pi/4}\vert 1 \rangle)$, can be taken to a set of states inside STAB:
\begin{equation}
    T^\dagger \{\vert 0\rangle \langle 0 \vert , \vert T \rangle  \langle T \vert \}T \mapsto \{\vert 0 \rangle \langle 0 \vert, \vert + \rangle \langle + \vert \} \subset \text{STAB}\,,
\end{equation}
where $T \coloneqq \text{diag}(1,e^{i\pi/4})$. Hence, this pair is \emph{not} set-magic, although it clearly has nonstabilizerness in the usual sense.

\textit{Nonstabilizerness in QPUs.}-- In any certification protocol, the assumptions behind the test play a central role. In our case, we assume: (a) multi-qubit systems and (b) the ability to estimate two-state overlaps. The exact way in which this estimation is obtained may be unknown. Common techniques involve performing SWAP tests, Bell measurements, or generating specific prepare-and-measure statistics. Minimizing sample and measurement complexity is desired for scalability, although not essential for the certification \emph{per se}. Commonly, a minimal requirement is using a number of samples and measurements smaller than the one needed to make (ideal) quantum state tomography~\cite{flammia2012quantum,yuen2023improvedsample, Haah2016, O_Donnell2016,chen2023does}, which is of order $O(d_T^2/\varepsilon^2)$, where $\varepsilon$ is a fixed precision with respect to distance functions and $d_T$ is the dimension of the whole Hilbert space considered. For single QPUs with $n$ qubits and associated Hilbert space $\mathcal{H}=(\mathbb{C}^2)^{\otimes n}$, a network of $s \in \mathbb{N}$ units has space $\mathcal{H}^{\otimes s}$, implying $d_T=\dim(\mathcal{H}^{\otimes s})=2^{sn}$. 

As shown in Appendix~\ref{app: contextuality}, event-graph inequalities cannot witness nonstabilizerness in general. This is somewhat unsurprising, as they were not proposed for such a task. However, we show that \emph{some} event-graph inequalities \emph{are} witnesses of both nonstabilizerness and coherence, starting with the cycle inequalities.
\begin{theorem}\label{theorem: cycles witness magic}
    Every cycle inequality violation $c_m(r)> m-2$ is a robust witness of nonstabilizerness for any set $\underline\rho$ of multi-qubit states such that $r=r(\underline\rho)$.
\end{theorem}

We defer the proof of this result to Appendix~\ref{app: proof of cycle results} but its underlying idea is extremely simple: After showing that pure states maximize any facet-defining inequality functional, we can use Theorem~\ref{theorem: quantization} to look only at some specific overlaps that, in principle, could violate the inequalities. We then show that any stabilizer realization satisfies $c_3(r(\underline\psi_{\text{STAB}})) \leq 1$ and use induction to show that the same must hold for all integers $m\geq 3$. Importantly, Theorem~\ref{theorem: cycles witness magic} holds for any set of multi-qubit states and not only for pure states.

Here, it is important to pause and clarify precisely the nature of the certifications allowed by Theorem~\ref{theorem: cycles witness magic}. Suppose we have a network with $s$ distinct units. Then, one possible task is to certify that each individual QPU can produce or has produced magic. This is illustrated in Fig.~\ref{fig: protocol}(a). Alternatively, we might be interested only in certifying that the network as a whole can generate magic \emph{somewhere} across its parties. In this case, we can distribute the state preparation and overlap estimation across different units (see Fig.~\ref{fig: protocol}(b)), reducing the number of overlaps that need to be evaluated in each QPU. Compared to other protocols~\cite{tirrito2023quantifying, turkeshi2023measuring} ours has the advantage that it foregoes the need to entangle the different parties in the network to carry out this certification.
\begin{figure}[t]
    \centering
    \includegraphics[width=0.9\columnwidth]{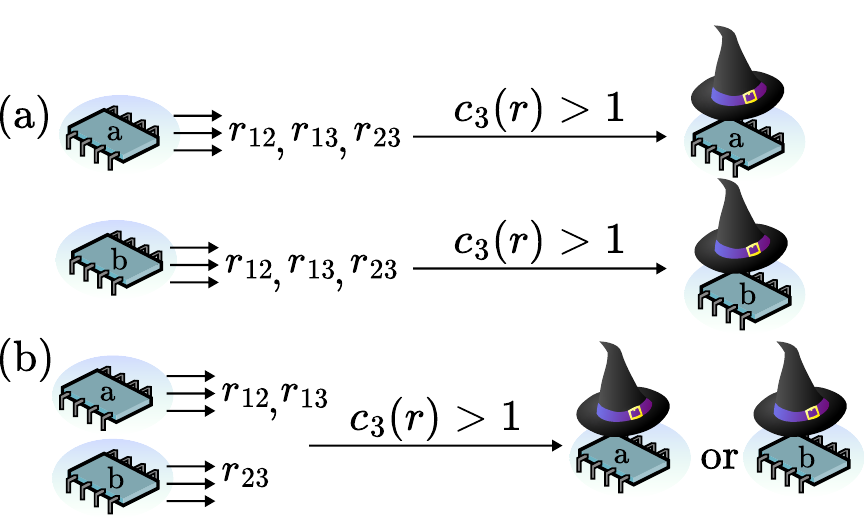}
    \caption{Certifying nonstabilizerness of quantum computing networks. (a) One can use our inequality witnesses for each device individually. (b) In case it is not necessary to certify that all devices have nonstabilizerness, but that \textit{some do}, it is possible to reduce the time usage by distributing the overlap computation. We use $c_3(r)$ as an example, but the same remains valid for $c_m(r)$ violations.
    }
    \label{fig: protocol}
\end{figure}

To make things even clearer, suppose that the certification of each unit in Fig.~\ref{fig: protocol}(a) requires the evaluation of $m$ overlaps. In total, the individual certification of \emph{all} units will require a total of $sm$ overlaps to be estimated. If we instead consider the distributed certification depicted in Fig.~\ref{fig: protocol}(b), the evaluation of $m$ overlaps suffices to certify the entire network for the presence of set magic. 

We take this opportunity to highlight two notable features of our work. First, suppose a single QPU in the network is capable of generating non-stabilizer resources. The argute reader will note that, rather than taking the distributed approach, by simply evaluating the overlaps on the proper QPU, one not only certifies the entire network for the presence ``\emph{somewhere}'', but further certifies the special unit. However, doing so requires prior knowledge of which QPU is capable of generating nonstabilizerness. If that is not the case, running the overlaps on a single (randomly chosen) QPU has a probability of success of $1/s$ whilst the distributed certification always succeeds. Secondly, we will see below that, in some cases, this same procedure can enable an even stronger certification, allowing us to guarantee that  \emph{all but one} QPU have produced non-stabilizer results.

We explained how to certify the production of magic (somewhere) in a given quantum hardware.  Alternatively, Theorem~\ref{theorem: cycles witness magic} can also be used for a second (subtly different) task. Suppose we are supplied with an ensemble of $m$ unknown states provided by third parties and are asked the question: Is the set magic? Then, we can use our network to answer this question by distributing the estimation of the overlaps among any number of the available QPUs. We inevitably need to evaluate $m$ overlaps to accomplish this task, regardless of how we choose to distribute this evaluation. Again, unlike previous protocols~\cite{tirrito2023quantifying, turkeshi2023measuring}, this has the added advantage of avoiding additional entanglement when certifying the presence of magic.

Note the two tasks are fundamentally different. Crucially, in the first, we want to \emph{certify} that magic can be produced within our hardware. In contrast, in the second, we aim to \emph{witness} the presence of magic in a set of states supplied by (any number of) third parties. In Appendix~\ref{app: comparison with others}, we make a comprehensive comparison between our scheme and other protocols.

From the known connection of event-graph inequalities to contextuality~\cite{budroni2022kochen,wagner2023inequalities}, the astute reader may question if such a result is in tension with or trivially follows from known results~\cite{budroni2022kochen}. In Appendix~\ref{app: contextuality}, we discuss how our work connects with contextuality theory and resolve these two points arguing that neither do our results trivially follow from known connections between contextuality and magic~\cite{schmid2022uniqueness,lillystone2019contextuality,howard2014contextuality}, nor are they in tension with them.

\textit{Witnessing nonstabilizerness in higher-dimensional systems.}-- Cycle inequalities certify the presence of nonstabilizerness but, as stated above, they are always maximally violated by sets $\underline\rho \subseteq \mathcal{D}(\mathbb{C}^2)$ of single-qubit states. This feature implies that such a certification scheme is not capable of capturing genuine properties of multi-qubit nonstabilizer resources, i.e., magic states in $\mathcal{D}(\mathbb{C}^{2^n})$. 

To address the possibility of witnessing nonstabilizerness that necessitates having access to higher Hilbert space dimensions, we study inequalities that witness \textit{both} of these properties. The simplest example of such inequalities is $h_4(r)\leq 1$ which requires two-qubit systems (or single qutrits) to have a violation, as shown in Ref.~\cite{giordani2023experimental}. We complement this result by showing that violations of $h_4(r)\leq 1$ also witness nonstabilizerness. 
\begin{theorem}\label{theorem: h4 is a magic witness}
    The inequality $h_4(r) \leq 1$ cannot be violated by quantum realizations $r=r(\underline\rho_{\mathrm{STAB}})$ of sets of stabilizer states $\underline\rho_{\mathrm{STAB}}$.
\end{theorem}

We defer the proof of this result to  Appendix~\ref{app: hm family}.

\textit{Robustness and scalability.}-- These inequalities, $c_m(r)\leq m-2$ and $h_4(r)\leq 1$, are our main nonstabilizerness witnesses. By default, they are both \emph{robust to incoherent noise} and \emph{scalable}. Robustness follows from the fact that they remain valid witnesses if, rather than pure states, we consider any set of states. Specifically, any state inside STAB  cannot violate these inequalities. On the other hand, scalability follows from the fact that (i) two-state overlaps are well-defined independently of any Hilbert space dimension, and (ii) quantum computers can efficiently estimate overlaps~\cite{fanizza2020beyond}. We refer to a protocol as scalable, or efficient, if one does not require exponentially increasing computational time, number of measurements, or samples.  Additionally, in our case, it is clear that scalability will also depend on the trade-off between dimension and number of overlaps estimated in a given inequality, for instance estimating $h_{2^n}(r)$ is certainly \emph{not} efficient. 

\textit{Full set magic.}-- We have seen that our inequality witnesses certify the presence of set magic for  $\underline\rho$. This indicates that \emph{some} state(s) in the set must lie outside of the STAB polytope. It is therefore natural to ask if it is possible to certify a lower bound on the number of states that must always lie outside of STAB. Clearly, for any set $\underline\rho$, at least one state $\rho_{\mathsf{ref}}$ can always be unitarily mapped inside the STAB. With this in mind, we introduce the following notion.
\begin{definition}[Full set magic] Let $\underline\rho \subseteq \mathcal{D}(\mathcal{H})$ be a finite set of states, $d=\dim(\mathcal{H})<\infty$. We say that the set $\underline\rho$ is fully set magic if, for every unitary $U:\mathcal{H}\to \mathcal{H}$, all states (but one) lie outside of the stabilizer polytope. 
\end{definition}

\begin{figure}[t]
    \centering
    \includegraphics[width=0.9\columnwidth]{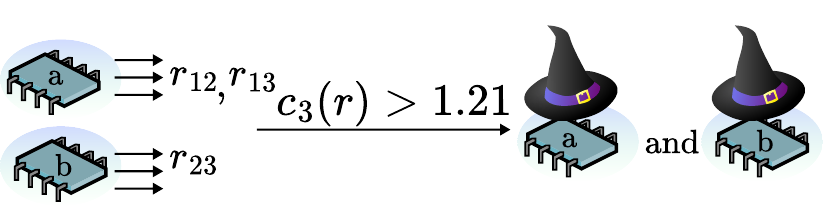}
    \caption{Full set magic certification. Fixing a certain target dimensionality (in this case single-qubit systems), one can certify all elements in a network from the same inequality values. Equivalently, one can certify various qubits in the same quantum computer running a parallel computation of the inequality values.}
    \label{fig: full set magic}
\end{figure}

There exist sets of states with full set magic. Let $\mathcal{H} = \mathbb{C}^2$ and consider any triplet $\{\rho_1,\rho_2,\rho_3\}$ of  states. It can be shown numerically that $c_3(r(\underline\rho))\leq 1.21$ if at most one state in $\underline\rho$ is outside of STAB. Recall that the maximal value of $c_3$  achieved with generic quantum states is $1.25$. Note that, for analyzing full set magic, we \emph{fix} the dimension. Therefore, there may be multi-qubit systems realizing values larger than $c_3(r) = 1.21$ while having only one magic state. In Appendix~\ref{app: full set magic}, we present numerical bounds for witnessing full set magic using other cycle inequalities and show numerical evidence that only odd cycle inequalities can witness full set magic. We also discuss concrete implementations of our certification scheme.

\textit{Discussion and further directions.}-- We have introduced novel witnesses of nonstabilizerness applicable to a single QPU and networks of QPUs. Our witnesses are robust in the sense of applying to pure and mixed states of any kind. They are also scalable, meaning they are independent of the system's dimension. This scalability, however, depends on the number of overlaps required to estimate a certain inequality. 

Various of the tools put forward either have their own technical interest or present novel theoretical opportunities. For instance, set magic and full set magic can be further investigated within the resource-theory framework. In this Letter, set magic arises naturally from using unitary invariants as our witnesses. There are several open questions on the possible operational importance of these results for quantum computation. Is it possible to connect set magic to the hardness of classically simulating a stabilizer circuit given a specific set of input states? Does a suitable \emph{quantifier} for set magic exist? Framing simulation within a unitary-invariant framework could better pinpoint the resources responsible for the exponential overhead of classical simulation and even lead to a unified view of different simulation schemes~\cite{masotllima2024stabilizer}.

In another direction, exploring if a similar effect to full set magic exists when considering the notion of set coherence could be of interest~\cite{designolle2021set,galvao2020quantum}. Moreover, for composite systems, large violations of the inequality $h_4(r) \leq 1$ may also require the presence of entanglement. We believe this is an intriguing aspect that merits further investigation.  Finally, we showed certain classes of inequalities that witness magic, but we do \emph{not} claim that these are the best ones. It is certainly possible that finding other overlap inequalities (or even going beyond two-state overlap to higher-order multivariate trace invariants) can result in better witnesses, simplifying experimental implementations.\\

We would like to thank Raman Choudhary for useful discussions. FCRP and RW acknowledge support from FCT -- Fundação para a Ciência e a Tecnologia (Portugal) through PhD Research Scholarship 2020.07245.BD, and PhD Grant SFRH/BD/151199/2021,  respectively. EZC acknowledges funding by FCT/MCTES through national funds and, when applicable, co-funding by EU funds under the project UIDB/50008/2020. EZC also acknowledges funding by FCT through project 2021.03707.CEECIND/CP1653/CT0002. EFG acknowledges support from FCT via project CEECINST/00062/2018. This work was supported by the ERC Advanced Grant QU-BOSS, GA no. 884676.

\bibliography{Bibliography}

\newpage

\noindent\textbf{Supplementary Material for: \textit{Certifying nonstabilizerness in quantum processors}}

\appendix

\tableofcontents

\section{Quantum bounds for non-stabilizer witnesses}\label{app: SDP}

\subsection{Lower bounds}

Let $f:\mathcal{D}(\mathcal{H})^{|V(\mathcal{G})|} \to \mathbb{R}$ be a convex-multilinear functional defined with respect to some facet-defining inequality of $\mathfrak{C}(\mathcal{G})$, for some $\mathcal{G}$, i.e., $f(\cdot) \equiv h(r(\cdot))$. In its generic form, $f$ is given by,
\begin{equation}
    f(\underline\rho) = \sum_{u,v \in V(\mathcal{G})} \alpha_{uv} \text{Tr}\left(\rho_u\rho_v\right),
\end{equation}
where $\alpha_{uv} \in \mathbb{Z}$ for all $u,v$. We are interested in evaluating the quantum bound of this functional, i.e., $\max_{\{\rho_x\}_{x=1}^N} f\,,$ over all possible quantum realizations, where $N = |V(\mathcal{G})|$.

To determine lower bounds for the quantum bound, we employ a seesaw of SDP. Indeed, by fixing all the states except one, the problem of maximizing the bound over the remaining state is an SDP.

\subsection{Upper bounds: hierarchy of semidefinite programming relaxations}

A way to evaluate upper bounds for this quantity is to adapt the Navascues-Vértesi (NV) hierarchy \cite{Navascues2015}  to this particular problem. In a similar spirit, we make a list of operators $\mathcal{S} = \{\mathds{1}, \{\rho_x\}\}$ and choose a degree of relaxation $k$. A relaxation of degree $k$ consists of keeping all products of at most $k$ operators from the list. The moment matrix is then constructed,
\begin{equation}
    \Gamma_{ij} = \text{Tr}\Big[S_i S_j\Big],
\end{equation}
where $S_i\in \mathcal{S}$. For a good introduction and review of SDP techniques in quantum information science see Ref.~\cite{tavakoli2023semidefinite}.

One samples a linearly independent basis of such moment matrices $\{\Gamma^{(1)},\dots, \Gamma^{(m)}\}$. The relaxation then consists of finding an affine combination $\Gamma = \sum_{j=1}^mc_j\Gamma^{(j)}\in f$, with $\Gamma\geq 0$, as follows,
\begin{equation}
\begin{aligned}
\max_{\vec{c}\in\mathbb{R}^m} \quad & f(\Gamma)\\
\textrm{s.t.} \quad & \Gamma \geq 0,\\
  & \sum_{j=1}^m c_j = 1.   \\
\end{aligned}
\end{equation}

\subsection{Maximal quantum violation of $m$-cycle inequalities}

\subsubsection{Qubit model}

The optimal solution of the maximization problem for $c_m(r)$ is a set of states $\{|\psi_x\rangle\}_{x=1}^m$ of the form
\begin{equation}\label{eq: states max cycles}
    |\psi_x\rangle = \cos (\theta_x^{(m)}) |0\rangle + \sin (\theta_x^{(m)})|1\rangle.
\end{equation}
where 
\begin{equation}
  \theta_x^{(m)} =
    \begin{cases}
      \frac{\pi}{2} - \frac{(x-1)\pi}{2m} & \text{if $m$ is odd,}\\
      \frac{\pi}{2} + \frac{(x-1)\pi}{2m} & \text{if $m$ is even.}
    \end{cases}       
\end{equation}

This family of optimal states is found explicitly using the SDP seesaw. Using the hierarchy, we could prove that the quantum bound is tight (to $10^{-5}$ precision) for this family of states up to $m = 8$. For larger values of $m$, $m \leq 20$, we have checked the lower bounds match the analytical formula up to the same precision. For  $m > 20$, we conjecture that this is the optimal quantum solution.

\subsubsection{Asymptotic behaviour}

Let us denote the set of optimal quantum states as $\underline\psi^{max}$. The maximum quantum violation is then given by
\begin{equation}
    c_m(r(\underline\psi^{max})) = (m-1)\cos^2\left(\frac{\pi}{2m}\right)-\cos^2\left(\left(1-\frac{1}{m}\right)\frac{\pi}{2}\right).
\end{equation}
On the other hand, the optimal incoherent quantum realization is $c_m(r(\underline\psi^{inc})) = m-2$. The fraction between these two quantities goes to $1$ as $m$ grows.
\begin{proposition}
    The limit of the ratio $c_m(r(\underline\psi^{inc}))/c_m(r(\underline\psi^{max}))$ when $m\to \infty$ is 1.
\end{proposition}

\begin{proof}
To see this, write
\begin{equation}
    \frac{c_m(r(\underline\psi^{inc}))}{c_m(r(\underline\psi^{max}))} = \frac{m-2}{(m-1)\cos^2(\frac{\pi}{2m})-\cos^2(\frac{1}{2}(1-\frac{1}{m})\pi)}
\end{equation}

The term $\cos^2(\frac{1}{2}(1-\frac{1}{m})\pi)$ is bounded, therefore the limit reduces to
\begin{align*}
    &\lim_{m\to\infty} \frac{c_m(r(\underline\psi^{inc}))}{c_m(r(\underline\psi^{max}))} = \lim_{m\to\infty} \frac{m-2}{(m-1)\cos^2(\frac{\pi}{2m})} = \\&\lim_{m\to\infty} \frac{m-2}{m-1} = 1.
\end{align*}
\end{proof}
\begin{figure}[t]
    \centering
    \includegraphics[width=\columnwidth]{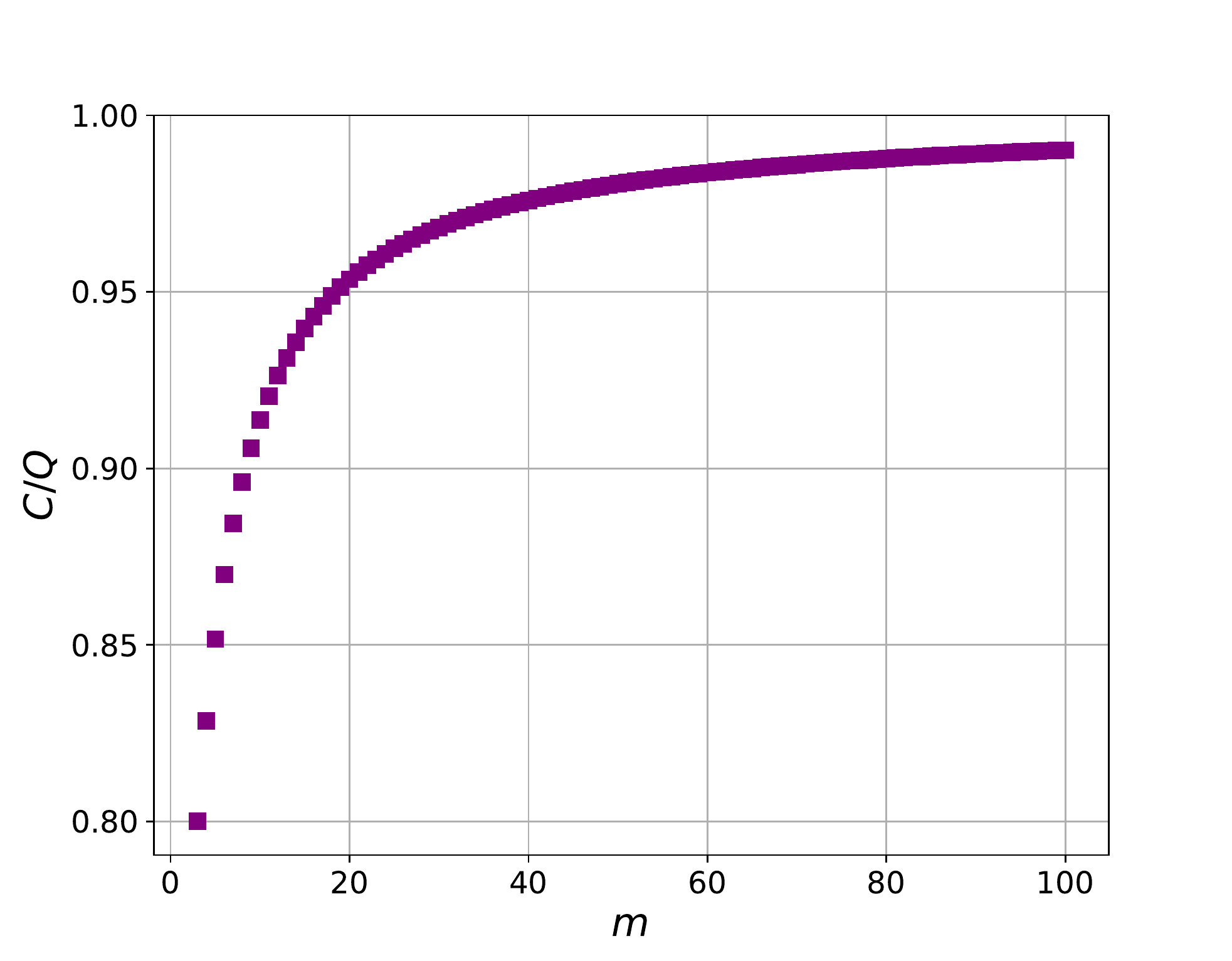}
    \caption{\textbf{Asymptotic behavior of $C/Q$ for the optimal solution of the cyclic inequalities.} (color online) Here, $C \equiv m-2$, which is the optimal value attainable by $c_m(r)$ with incoherent quantum realizations $r=r(\underline\rho^{inc})$. $Q$ represents the optimal value obtainable using any quantum state, attained with $\underline\rho = \underline\psi^{max}$ given by Eq.~\eqref{eq: states max cycles}. We present their fraction as a (continuous) function of $m$.}
    \label{fig:asymptotic_cyclic}
\end{figure}
Fig.~\ref{fig:asymptotic_cyclic} shows the ratio as a function of $m$, denoted, for simplicity, as $C/Q$. The numerics above suggest that pure single-qubit states always provide the optimal violation. In the following, we prove this to be true.

\subsubsection{Qubit optimality}

\begin{proof}[Proof of Theorem~\ref{theorem: qubits cycle attainable}]
    Let $\mathcal{C}_m$ be the $m$-cycle event graph and $r$ an edge-weight. We show later (see Lemma~\ref{lemma: pure_state upper_bound} of Appendix~\ref{app: proof of cycle results}) that, for any quantum realization $r=r(\underline\rho)$, there exists some pure state quantum realization $r=r(\underline\psi)$ for which $c_m(r(\underline\rho)) \leq c_m(r(\underline\psi))$. Therefore, for optimal violations, we can focus on pure state realizations. Let $r=r(\underline\psi)$ be any pure state quantum realization for an edge-weight $r$ of the event graph $\mathcal{C}_m$. Then, consider the unitary sending $\underline\psi = \{\vert \psi_1\rangle, \vert \psi_2\rangle , \dots, \vert \psi_m\rangle \}$ to $\underline{\tilde{\psi}} = \{\vert 0\rangle, \vert \Tilde{\psi}_2\rangle, \dots, \vert \Tilde{\psi}_m\rangle \}$. In this case, any overlap $r_{i,i+1}$ with respect to elements in $\underline{\tilde{\psi}}$ in the cycle will be equal to $\cos^2(\theta_{i+1}-\theta_{i})$. The angle $\theta_i$ is the angle between the vectors $\vert \tilde{\psi}_i\rangle$ with the reference state $\vert 0\rangle$, and can have any value $\theta_i \in [0,\pi)$, for all $i=1,\dots,m$. The above implies that we can write any $c_m(r(\underline\psi))$ as
    \begin{align*}
        c_m(r(\underline\psi))=&-r_{1,2}+r_{2,3}+\dots +r_{m-1,m} + r_{m,1} \\
        = &-\cos^2(\theta_2)+\cos^2(\theta_3-\theta_2)\\+&\dots+\cos^2(\theta_{m}-\theta_{m-1})+\cos^2(\theta_m).
    \end{align*}
    Note that this is independent of any dimensionality constraint. Nevertheless, each overlap in the tuple $(r_e)_{e \in \mathcal{C}_m}$ attaining the values above can also be obtained by the set of states
    \begin{align*}
        \vert q_1\rangle &= \vert 0\rangle, \\
        \vert q_2\rangle &= \cos(\theta_2)\vert 0\rangle + \sin(\theta_2)\vert 1\rangle \,, \\
        \vert q_3\rangle &= \cos(\theta_3)\vert 0\rangle + \sin(\theta_3)\vert 1\rangle \,, \\
        &\vdots \\
        \vert q_m\rangle &= \cos(\theta_m)\vert 0\rangle + \sin(\theta_m)\vert 1\rangle \,,
    \end{align*}
    that is, within the space $\mathbb{C}^2$. Hence, using just the single-qubit states $\underline q= \{\vert q_x\rangle\}_{x=1}^m$, we have that $c_m(r(\underline\psi)) = c_m(r(\underline q))$. In other words, \emph{any} value attained by $c_m(r)$ with pure quantum states of any dimension is also attained by using just single-qubit states. In particular, this will also be valid for the \emph{maximal} values of $c_m(r)$. This concludes the proof.
\end{proof}

\subsection{Maximum quantum violation of $h_4(r) \leq 1$} 

Again, using the methods described in this Appendix, we can show tight bounds for the maximal quantum realizations violating $h_4(r)\leq 1$, for sets $\underline\rho \subseteq \mathcal{D}(\mathbb{C}^d)$ with $d = 3, 4$. 

Note that $h_4$ (as any event-graph functional) satisfies $h_4(r(\underline\rho)) = h_4(r(\underline{U \rho U^\dagger }))$ for any unitary $U$ and $\underline{U \rho U^\dagger} \equiv \{U \rho_i U^\dagger \}_{i=1}^m$. This implies that any maximal quantum violation can be attained using Hilbert spaces of at most dimension $d=4$. We find a tight bound of $1+1/3$, up to $10^{-5}$ precision, for both $d = 3, 4$, proving that qutrits are sufficient to achieve the maximum quantum violation of $h_4$.

\section{Proof of Theorem~\ref{theorem: cycles witness magic}}\label{app: proof of cycle results}

Before proving  Theorem~\ref{theorem: cycles witness magic}, we present some novel technical results that are of general interest to the event-graph formalism introduced in Ref.~\cite{wagner2023inequalities}. 

Following the same notation used in the main text, for any given event graph $\mathcal{G}$, with edges $E(\mathcal{G})$ and vertices $V(\mathcal{G})$, we denote the polytope of all edge-weights $r: E(\mathcal{G}) \to [0,1]$ realizable by set-incoherent tuples as $\mathfrak{C}(\mathcal{G})$. The Hamming weight of a string (or, equivalently, a tuple) $s$ of 0/1-assignments equals the number of $1$ assignments in the string (or tuple) and is denoted $|s|_{\text{H}}$.
\begin{lemma}\label{lemma: equal_no_violation}
    Let $\mathcal G=\mathcal C_m$ be the $m$-cycle graph. Fix $m\geq 3$ and some $\tilde{e} \in E(\mathcal{C}_m)$. If we denote $$\mathfrak{C}_{\tilde{e}} := \{r \in \mathfrak{C}(\mathcal{C}_m)\,:\, r_{\tilde{e}}=1\},$$ we have that $\mathfrak{C}_{\tilde{e}} = \{1\} \times  \mathfrak{C}(\mathcal{C}_{m-1}) $. 
\end{lemma}

This lemma shows that if we define the cross-section $\mathfrak{C}_{\tilde{e}}$ of the polytope $\mathfrak{C}(\mathcal{C}_m)$ along the direction $r_{\tilde{e}}=1$, the resulting polytope is isomorphic to $\mathfrak{C}(\mathcal{C}_{m-1})$. This implies that given any facet-defining inequality of $\mathfrak{C}(\mathcal{C}_m)$, if we force an edge to be equal to one, the resulting inequality will be a facet-defining inequality of $\mathfrak{C}(\mathcal{C}_{m-1})$, or a trivial inequality. 

\begin{proof}[Proof of Lemma~\ref{lemma: equal_no_violation}]
    Because $\mathfrak{C}(\mathcal{C}_m)$ is a convex polytope, the same is true for $\mathfrak{C}_{\tilde{e}}$. Let $\text{ext}(P)$ denote the set of extremal points of any convex polytope $P$, and hence $P = \textsc{ConvHull}[\text{ext}(P)]$. Let us assume, wlog, an ordering $r=(r_e)_{e \in E(\mathcal{C}_m)} = (r_{\tilde{e}},r_{e_1},\dots,r_{e_{m-1}})$. We have, 
    $$\text{ext}(\mathfrak{C}_{\tilde{e}}) = \{(1,s)\in \mathbb{R}^m: s \in \text{ext}(\mathfrak{C}(\mathcal{C}_{m-1}))\}.$$
    
    ($\subseteq$) Let $\tilde{r} \in \text{ext}(\mathfrak{C}_{\tilde{e}})$. By construction, we must have that $\tilde{r} \equiv (1,\tilde{s}) \in \text{ext}(\mathfrak{C}(\mathcal{C}_m))$, with $\tilde{s}$ a deterministic assignment for which $|\tilde{s}|_{\text{H}} \neq m-2$. Therefore, $\tilde{s} \in \text{ext}(\mathfrak{C}(\mathcal{C}_{m-1}))$.
    
    ($\supseteq$) This direction follows trivially.

    Hence, we have that
    \begin{equation}
        \mathfrak{C}_{\tilde{e}} = \{1\} \times \mathfrak{C}(\mathcal{C}_{m-1}),
    \end{equation}
    where $\{1\}$ is the singleton polytope, as we wanted. 
\end{proof}

This simple result will be instrumental for constructing the inductive step used in the proof of Theorem~\ref{theorem: cycles witness magic}.

We now make the notion of a quantum realization within the context of the event graph formalism precise. We can associate nodes of the graph to quantum states $V(\mathcal G) \ni v \stackrel{\ell}{\mapsto} \rho_v \in \underline\rho$ via some vertex $\underline\rho$-labeling $\ell: V(\mathcal{G}) \to \underline\rho$. Once such labeling $\ell$ is fixed, we associate edge-weights $r_e \equiv r_{u,v}$ to two-state overlaps $E(\mathcal G) \ni e =\{u,v\} \stackrel{\ell}{\mapsto} \{\rho_u,\rho_v\} \stackrel{\langle \cdot , \cdot \rangle_{HS}}{\mapsto} \text{Tr}(\rho_u\rho_v)$~\footnote{Sometimes the inner-product $\langle \phi_i | \phi_j \rangle $ is called an overlap. We will not use this terminology here and simply refer to the overlap as the absolute value square of inner-product between states, or more generally, to the trace between the product of two density matrices.}, with $\langle X, Y \rangle_{HS} = \text{Tr}(X^\dagger Y)$ the Hilbert-Schmidt inner product. 

The cardinality of $\underline\rho$ is not necessarily equal to that of $V(\mathcal{G})$, e.g., the same state can be associated to all vertices by the labeling $\ell(v) = \rho, \forall v \in V(\mathcal{G})$. Each vertex labeling $\ell$ is isomorphic to a tuple $(\ell(v))_{v \in V(\mathcal{G})} \in \mathcal{D}(\mathcal{H})^{|V(\mathcal{G})|}$. Given some $\underline\rho$-labeling $\ell$, we can see $r$ as a function that outputs a tuple of two-state overlaps $r_\ell(\underline\rho)$ for an input set $\underline\rho$. For instance, take $\underline\rho = \{\vert \psi \rangle \langle \psi \vert\} \cup \{\sigma_1,\sigma_2\}$ and $\ell(v)=\vert \psi \rangle \langle \psi \vert$ for all $v \in V(\mathcal{G})$, as above. Since all vertices of $V(\mathcal{G})$ have been assigned the same state, the associated $r_\ell(\underline\rho)$ is  $r_\ell(\underline\rho) = (1,1,\dots,1)$. When it is clear from the context which labeling $\ell$ is being used, we simply write $r(\underline\rho)$.

Recall that, any edge-weight $r: E(\mathcal{G}) \to [0,1]$ for an event graph $\mathcal{G}$, is said to have a quantum realization~\cite{fraser2023realization,wagner2024coherence} if there exists $\underline\rho = \{\rho_i\}_{i \in V(\mathcal{G})}$ such that $r=r(\underline\rho) \equiv (\text{Tr}(\rho_i\rho_j))_{\{i,j\} \in E(\mathcal{G})}$. We denote $|X|$ the cardinality of a set $X$. We say that a quantum state $\rho \in \mathcal{D}(\mathcal{H})$ is \emph{pure} when $\text{Tr}(\rho^2) = 1$, in which case we denote it as a rank-1 projector $\rho=\vert \psi \rangle \langle \psi \vert \equiv \psi$.
\begin{lemma}\label{lemma: pure_state upper_bound}
    Let $h:\mathbb{R}^{|E(\mathcal{G})|}\to \mathbb{R}$ be any convex-linear functional, acting over elements $r \in [0,1]^{|E(\mathcal{G})|}$, for any event graph $\mathcal{G}$. Then, for any quantum realization $r=r(\underline\rho)$ with states $\{\rho_i\}_{i}$, there exists a pure state quantum realization $r=r(\underline\psi)$ with states $\{\vert \psi_i \rangle\}_{i}$, such that \begin{equation}
        h(r(\underline\rho)) \leq h(r(\underline\psi)).
    \end{equation}
    Moreover, $\underline\psi \subseteq \textsc{ConvHull}(\underline\rho)$.
\end{lemma}

\begin{proof}
    Each $\rho_i \in \underline\rho$ is a convex combination of pure states $\{\psi^{(i)}_{\omega_i}\}_{\omega_i \in \Omega_i}$ for some ensemble of pure states $\Omega_i$. Noticing that $h(r)$ are, by construction, linear functionals over the overlaps,
    \begin{align*}
        \forall i,\rho_i = \sum_\omega \lambda_\omega^{(i)}\vert \psi_\omega^{(i)} \rangle \langle \psi_\omega^{(i)} \vert \Rightarrow\hspace{1.3cm}\\
        h(r(\{\rho_i\}_i)) = \sum_{\omega_1,\dots,\omega_m} \lambda_{\omega_1}^{(1)}\dots \lambda_{\omega_m}^{(m)}h(r(\{\psi_{\omega_i}^{(i)}\}_i)). 
    \end{align*}
    To conclude the above, one needs to introduce some redundant values of $1 = \sum_{\omega_i}\lambda_{\omega_i}^{(i)}$. The equation then follows from linearity with respect to $r$, and hence multilinearity with respect to the states. 
    
    We can collectively write $s = (\omega_1,\dots,\omega_m)$ and define $q_s = \lambda_{\omega_1}^{(1)}\dots \lambda_{\omega_m}^{(m)}$. Because each set of weights $\{\lambda_{\omega_i}^{(i)}\}_{\omega_i \in \Omega_i}$ correspond to convex weights, i.e., $\sum_{\omega_i}\lambda_{\omega_i}^{(i)} =1 $ with $0\leq \lambda_{\omega_i} \leq 1$, we get that $\{q_s\}_s$ is also a set of convex weights. With this simplified notation we have that 
    $h(r(\{\rho_i\}_i)) = \sum_sq_s h(r(\{\psi^{(i)}_s\}_i))$ with $\sum_sq_s=1$ and $0\leq q_s\leq 1$. In other words, the linear functional $h$ realized by overlaps between general quantum states can be written as the convex combination of the same functional realized by overlaps between pure states. Choosing now a particular $s^\star$ such that $\forall s, h(r(\{\psi_{s^\star}^{(i)}\}_i)) \geq h(r(\{\psi_s^{(i)}\}_i))$ we see that $$h(r(\{\rho_i\}_i)) = \sum_sq_s h(r(\{\psi^{(i)}_s\}_i))\leq \sum_sq_s h(r(\{\psi_{s^\star}^{(i)}\}_i)).$$
    Since $\sum_sq_s=1$ we have that $h(r(\{\rho_i\}_i)) \leq h(r(\{\psi_{s^\star}^{(i)}\}_i))$.
\end{proof}

\begin{theorem}\label{theorem: functional max}
    Let $h(r)$ be any linear-functional over $r=(r_e)_{e \in E(\mathcal{G})}$ for any event graph $\mathcal{G}$ and $\mathfrak{Q}(\mathcal{G})$, defined by $$\mathfrak{Q}(\mathcal{G}):= \{r:E(\mathcal{G}) \to [0,1]\,:\, \exists \underline \rho , r=r(\underline\rho)\},$$ be the set of quantum realizable edge-weights. Then, there always exists a pure state quantum realization $r=r(\underline{\psi})$ such that  
    \begin{equation}
        h(r(\underline\psi)) = \max_{r \in \mathfrak{Q}(\mathcal{G})}h(r).
    \end{equation}
    The same holds if we restrict realizations to some convex and compact subset $\mathfrak{S} \subseteq \mathcal{D}(\mathcal{H})$ of all states, so that the quantum realizations are such that $r=r(\underline\rho_{\mathfrak{S}})$, with $\underline\rho_{\mathfrak{S}} \subseteq \mathfrak{S}$. 
\end{theorem}

\begin{proof}
    For any such $h$, Lemma~\ref{lemma: pure_state upper_bound} shows that to every quantum realization $r=r(\underline\rho)$, there exists a larger pure state realization within the convex hull of $\underline\rho$. Therefore, the maximum attainable value, among all quantum realizations $r \in \mathfrak{Q}(\mathcal{G})$, must be pure-state realizable, otherwise this would contradict Lemma~\ref{lemma: pure_state upper_bound}. The argument is the same if $\mathfrak{S}$ is used instead.
\end{proof}

This Theorem implies the immediate corollary.
\begin{corollary}
    The maximal quantum violation of any facet-defining inequality of $\mathfrak{C}(\mathcal{G})$, for any event graph $\mathcal{G}$, is attained by pure states.
\end{corollary}

\begin{proof}
    Since $\mathfrak{C}(\mathcal{G})$ is a convex polytope, any facet-defining inequality from $\mathfrak{C}(\mathcal{G})$ is described by convex-linear functionals $h(r)$, together with some $b \in \mathbb{R}$ satisfying $h(r) \leq b$.
\end{proof}

This Corollary proves (and generalizes) a conjecture from Ref.~\cite{galvao2020quantum}, that the maximal bounds violating the $c_3(r)\leq 1$ inequality using pure states were also valid for mixed states in general, and for any dimension. While here we will use these results to prove Theorem~\ref{theorem: cycles witness magic}, they are important by themselves for the event-graph approach and the theory of coherence witnesses. 

Finally, we prove Theorem~\ref{theorem: cycles witness magic} of the main text.

\begin{proof}[Proof of Theorem~\ref{theorem: cycles witness magic}]
    Due to Theorem~\ref{theorem: functional max}, we can restrict ourselves to pure stabilizer states. Consider first the $3$-cycle inequality. We have
    \begin{equation*}
        r_{1,2}+r_{1,3}-r_{2,3}\leq 1.
    \end{equation*}
    From Theorem~\ref{theorem: quantization}, we see that if all states in the graph are oblique to their neighbors there can be no violation since $r_{1,2}+r_{1,3}-r_{2,3} \leq r_{1,2}+r_{1,3} \leq \frac{1}{2^k}+\frac{1}{2^{k'}} \leq 1$ for all $k,k'=1,\dots,n$. The same holds if we allow some edge-weights to be zero. If we allow any edge-weight in the inequality to be equal to one, it is simple to see that we cannot have a violation, as we would have two nodes corresponding to the same stabilizer state, implying that the remaining pair of overlaps is equal. This shows the result for $c_3(r)$. 

    To show that the same is true for any $m$-cycle inequality we proceed by induction. Assume that an $m$-cycle inequality \emph{cannot} be violated by quantum realizations $r=r(\underline\rho_{\text{STAB}})$, with $\underline\rho_{\text{STAB}} \subseteq \text{STAB}$. For any $(m+1)$-cycle inequality we have that,  $\forall e \in E(\mathcal{C}_{m+1})$, 
    \begin{equation*}
        -r_e + \sum_{{\begin{array}{c} e' \in E(\mathcal{C}_{m+1}) \\e'\neq e\end{array}}}r_{e'} \leq \sum_{{\begin{array}{c} e' \in E(\mathcal{C}_{m+1}) \\e'\neq e\end{array}}}r_{e'} \leq \frac{m}{2}
    \end{equation*}
    for any set of pure stabilizer states oblique or orthogonal to their neighbors in the graph. Since $m/2 \leq m-2$ for all $m\geq 4$ it remains to show that if two (or more) neighboring stabilizer states are equal we still cannot have a violation. 
    
    From Lemma~\ref{lemma: equal_no_violation}, if any edge-weight is equal to one, this implies that the cycle inequality from $\mathfrak{C}(\mathcal{C}_{m+1})$ becomes an inequality from $\mathfrak{C}(\mathcal C_{m})$, which, by hypothesis, cannot be violated with the stabilizer subtheory, i.e., by any quantum realization $r=r(\underline\rho_{\text{STAB}})$. We conclude that if the stabilizer subtheory cannot violate inequalities from $\mathfrak{C}(\mathcal{C}_m)$ it will also not violate the inequalities from $\mathfrak{C}(\mathcal{C}_{m+1})$. As we know this is true for the cycle inequalities $\mathfrak{C}(\mathcal{C}_3)$, by induction, this property must be satisfied by all facet-defining inequalities for the event graph polytopes $\mathfrak{C}(\mathcal{C}_m)$ for any $m$.
\end{proof}

\section{Known facts regarding the relation between contextuality and magic}\label{app: contextuality}

In this Appendix, we will prove the existence of an event-graph inequality that is both facet-defining and violated by sets of stabilizer states. To do so, we use some results within the field of Kochen-Specker noncontextuality~\cite{budroni2022kochen}. We then extend these considerations to generalized noncontextuality~\cite{spekkens2005contextuality}.

To provide some context, recall that magic-state injection is the leading model for experimentally realizing fault-tolerant quantum computation. While it involves only stabilizer operations at every step of the computation, the injection of magic states elevates the model to quantum universality.  Ref.~\cite{howard2014contextuality} showed that contextuality is a necessary resource for universal quantum computation via magic-state injection. The scope of this result depends on whether the model involves even-prime dimensional qudits (i.e. qubits) or odd-prime qudits. For the latter case, a state is non-contextual if and only if it belongs to the set of states unable to unlock any computational speed-up. This set forms a polytope, denoted as $\mathcal{P}_{\text{SIM}}$, meaning that the subtheory within this polytope is efficiently simulable with classical computation. This polytope strictly contains the set of stabilizer states but is not equivalent to it. 

Therefore, for odd-prime dimensions, Kochen-Specker noncontextuality inequalities serve as witnesses of nonstabilizerness. However, for even-dimensional systems, the same does \emph{not} hold, as we now show. We start by constructing the relevant event graph. First, we construct the so-called \emph{exclusivity graph} $\mathcal{G}_{\text{exc}}$~\cite{cabello2014graph,amaral2018graph}. We take this graph to be the complement graph~\footnote{The complement of a graph $\mathcal{G}$ is a new graph $\mathcal{G}^{c}$ such that $V(\mathcal{G}) = V(\mathcal{G}^c)$ and, $e \in \mathcal{G}^c$ iff $e \notin \mathcal{G}$.} of the Shrikhande graph~\cite{shrikhande1958graph}. See Ref.~\cite[Fig.~2, pg.~10]{bharti2022graph} for a representation of  $\mathcal{G}_{\text{exc}}$. We follow closely the discussion of the proof of KS-contextuality discussed in Ref.~\cite{bharti2022graph}.  Secondly, we take the suspension graph~\cite[Def.~2.23, pg.~36]{amaral2018graph} $\nabla \mathcal{G}_{\text{exc}}$ by a new node $\star$. This new graph will be our event graph $\mathcal{G} = \nabla \mathcal{G}_{\text{exc}}$. It was shown in Ref.~\cite{wagner2023inequalities} that any inequality from an exclusivity graph $\mathcal{G}_{\text{ext}}$ is mapped to some facet defining inequality of the event graph $\mathcal{G}= \nabla \mathcal{G}_{\text{exc}}$. Therefore, the inequality 
\begin{equation*}
    \sum_{v \in V(\mathcal{G}_{\text{exc}})}r_{\star,v} \leq 3
\end{equation*}
is both a noncontextuality inequality (within the Cabello-Severini-Winter  framework~\cite{cabello2014graph}) and a facet-defining event-graph inequality, when a specific mapping takes place (see Ref.~\cite{wagner2023inequalities} for details). This inequality corresponds to Mermin's Bell inequality~\cite{mermin1990extreme} and can be violated by letting the vertices $v\in V(\mathcal{G}_{\text{exc}})$ be given by the stabilizer (separable) states
\begin{align}
    &\vert 0,+,+\rangle,  \vert 1,-,+\rangle, \vert 1,+,-\rangle, \vert 0,-,-\rangle\nonumber \\
    &\vert +,0,+\rangle, \vert -,1,+\rangle, \vert -,0,-\rangle, \vert +,1,-\rangle\label{eq: sep vectors CSW}\\
    &\vert +,+,0\rangle, \vert -,-,0\rangle, \vert -,+,1\rangle, \vert +,-,1\rangle \nonumber\\
    &\vert 1,1,1\rangle, \,\,\,\,\vert 0,0,1\rangle,\,\,\,\, \vert 0,1,0\rangle, \,\,\,\,\vert 1,0,0\rangle.\nonumber
\end{align}
and $\star$ by the Greenberger–Horne–Zeilinger (GHZ) state $\vert \mathrm{GHZ}\rangle = \frac{1}{\sqrt{2}}(\vert 0,0,0\rangle + \vert 1,1,1\rangle)$. In this way, 
\begin{equation*}
    \sum_{v \in V(\mathcal{G}_{\text{exc}})}r_{\star,v} = \sum_{v \in V(\mathcal{G}_{\text{exc}})} |\langle \mathrm{GHZ}| v\rangle |^2 = 4 > 3.
\end{equation*}

The above shows that there are event-graph inequalities that can be violated by stabilizer states. 

Let us now discuss the relationship between \emph{generalized noncontextuality}~\cite{spekkens2005contextuality} and magic. It was shown in Refs.~\cite{lillystone2019contextuality,schmid2022uniqueness} that odd-dimensional stabilizer subtheory allows for a generalized noncontextual model. In this case, any violation of a noncontextuality inequality attesting to the failure of generalized contextuality will also be a witness of having states (transformations, measurement effects) outside the stabilizer subtheory. However, similarly to the case of KS-noncontextuality, for even-dimensional systems, no such noncontextual model for the stabilizer subtheory exists. Therefore, in general, even dimensional stabilizer subtheory \emph{can} violate generalized noncontextuality inequalities. Finding which inequalities \emph{are not} violated by such stabilizer subtheory becomes a case-by-case study.

In summary, the connection between event graph inequalities and noncontextuality inequalities does not imply that any event graph inequality will immediately be a nonstabilizerness witness and therefore does \emph{not} render our results trivial. On the other hand, it also does not make our results immediately incorrect. The fact that our inequality witnesses (which can also be interpreted as noncontextuality inequalities) cannot be violated by stabilizer states does \emph{not} imply that they do not violate \emph{some} noncontextuality inequality. As it is in general, for a given KS-measurement scenario (or equivalently, a generalized prepare-and-measure noncontextuality scenario), it is only if one satisfies \emph{all} the noncontextuality inequalities that a noncontextual model exists.

\section{Comparison with other schemes}\label{app: comparison with others}

In this Appendix, we make a comprehensive review of existing methods for witnessing the nonstabilizerness of quantum states. Clearly, any \emph{quantification} scheme also constitutes a witness, thus, for a broader comparison, we were careful to include such methods as well. 

\subsection{Methods that require additional entanglement generation}

We start by presenting two witnessing schemes that require additional entanglement; they use multifractal flatness~\cite{turkeshi2023measuring} and spectral flatness~\cite{tirrito2023quantifying}. First, define the inverse participation ratio,
\begin{equation}
    I_q(\vert \psi \rangle) := \sum_{b \in \mathbb{F}_2^n}|\langle b|\psi\rangle|^{2q} = \sum_{b \in \mathbb{F}_2^n}r_{b,\psi}^q\,.
\end{equation}
We note that, to calculate $I_q(\vert \psi \rangle)$ for any fixed $q$, one needs to evaluate $d=2^n$ overlaps. The \emph{multifractal flatness} is defined as
\begin{equation}
    \mathcal{F}_{\text{multi}}(\vert \psi \rangle) := I_3(\ket{\psi})-(I_2(\ket{\psi}))^2\,.
\end{equation}
This quantity witnesses nonstabilizerness whenever we obtain $\mathcal{F}_{\text{multi}}(C\vert \psi \rangle) > 0$, where $C$ is an $n$-qubit Clifford operation. Since both calculating and measuring this quantity require $O(d=2^n)$ overlaps, this witnessing process is inefficient. We remark that, when averaged over the Clifford orbit of the state $\vert \psi \rangle$, this witness provides information about the stabilizer Rényi entropy $M_2$ which we will discuss in more detail later on. 

Three key aspects of this witness stand in stark contrast to our scheme. First, it is device-\emph{dependent}. Secondly, it is only applicable to pure states. Finally, it requires additional entangling gates to be applied over the state $\vert \psi \rangle$, stemming from the Clifford unitary, $C$, needed for $\mathcal{F}_{\text{multi}}(C\vert \psi \rangle) > 0$ to properly witness nonstabilizerness of $\vert \psi \rangle$.

Another function that can be used to witness nonstabilizerness is the \emph{entanglement spectrum flatness} $\mathcal{F}_A(\vert \psi \rangle)$. We consider a pure state $\vert \psi \rangle$~\footnote{Ref.~\cite{tirrito2023quantifying} exemplifies the task using an $n$-qubit fully separable state, but, to the best of our understanding, their protocol works for any state.} and some $n$-qubit Clifford operation $C$ such that $C \vert \psi\rangle$ is sufficiently entangled. In some cases, shallow Clifford evolutions are enough. We then choose an arbitrary bipartition $\mathcal{H}_A \otimes \mathcal{H}_B \simeq \mathbb{C}^{\otimes n}$ of the $n$-qubit system and calculate $\rho_A := \text{Tr}_B(C\vert \psi \rangle \langle \psi \vert C^\dagger)$. The entanglement spectrum flatness of this bipartition is given by 
\begin{equation}
    \mathcal{F}_{A}(C\vert\psi\rangle) \coloneqq \text{Tr}(\rho_A^3)-(\text{Tr}(\rho_A^2))^2.
\end{equation} 
If we obtain that $\mathcal{F}_A(C\vert \psi \rangle)>0$, $\vert \psi \rangle$ must be a magic state, and therefore $\mathcal{F}_A(C\vert \psi \rangle)>0$ acts as a witness of nonstabilizerness. This witness can be efficiently measured (in terms of the number of measurements and samples of $\rho_A$ required) using simple quantum circuits~\cite{wagner2024quantum,quek2024multivariatetrace,oszmaniec2021measuring}. 

Let us discuss more explicitly the advantages and drawbacks of these techniques when compared to the scheme we propose in the main text. Entanglement spectrum flatness has the advantage that it can witness almost any magic state (in the sense of Haar random states). In contrast, in our inequality-based witness, certifications should target certain overlap values that become increasingly rare for large systems to randomly access. This happens since the two-state overlap between Haar random $n$-qubit states behaves as $r_{i,j} \sim 1/2^n$. Moreover, both flatness results from above can approximate values of nonstabilizer monotones; thus, they can ultimately be used for quantification (a task strictly more powerful than witnessing). At the moment, our protocol has no known link with quantification tools, although we believe this to be an interesting direction for future research.

On the other hand, unlike our scheme, both of these methods are device-dependent and applicable only to pure states. Additionally, let us assume that we would like to certify the generation of magic in \emph{some} QPU of a network. Fig.~\ref{fig: protocol}(b) in the main text illustrates how to handle this task within our protocol. Notably, as explained therein, the certification can be distributed requiring fewer resources than if we were to certify each QPU individually [Fig.~\ref{fig: protocol}(a)]. Contrastingly, with either of the two witnesses presented here, distributing the certification requires us to entangle the degrees of freedom of the multiple QPUs in the network we are interested in certifying, due to the Clifford operations that must be applied. 

\subsection{Methods that require full information of the STAB polytope, or full information of the quantum state}

The vast majority of quantification schemes require full information on the stabilizer polytope. Beyond that, they often also require full (tomographic) information of the quantum state. The following monotones require complete knowledge of the underlying state \emph{and} of the STAB polytope: stabilizer fidelity~\cite{bravyi2019simulationofquantum}, stabilizer extent~\cite{bravyi2019simulationofquantum}, stabilizer rank~\cite{bravyi2016trading,bravyi2019simulationofquantum,bravyi2016improved}, stabilizer nullity~\cite{beverland2020lowerbounds}. Some that are also well-defined for generic mixed states, having the same drawbacks, are mana~\cite{veitch2014resource,hakop2015estimating}, all variations of the robustness of magic~\cite{seddon2021quantifying}, relative entropy~\cite{veitch2014resource}, min- and max-relative entropies~\cite{liu2022many}, and the dyadic negativity~\cite{seddon2021quantifying}.

It is interesting to remark that the stabilizer extent $\xi(\vert \Psi \rangle)$ has the extremely useful property of being multiplicative,
\begin{equation*}
    \xi(\ket{\Psi}) \coloneqq \xi(\vert \psi_1\rangle \otimes \dots \otimes \vert \psi_m \rangle) = \prod_{j=1}^m \xi(\vert \psi_j \rangle)\,,
\end{equation*}
provided that all states $\vert \psi_j \rangle$ are 1-, 2- or 3-qubit states, i.e., that $\underline\psi \subseteq \mathcal{D}(\mathbb{C}^{2^s})$ with $s \in \{1,2,3\}$~\cite{bravyi2019simulationofquantum}.

This implies that an alternative strategy to witnessing magic is to make 1-, 2-, or 3-qubit state tomography of all states $\ket{\psi_j}$, use that information to calculate their stabilizer extent, and then multiply the results. Compared to our scheme, beyond being significantly device-dependent and demanding great control of the system (since one must perform full tomography), our scheme requires a smaller number of measurements and samples (since overlap estimation is experimentally less demanding than performing full tomography, even for single-qubit systems). It has been shown that the stabilizer extent is \emph{not} multiplicative in general~\cite{heimendahl2021stabilizerextentis}. It is also clear that, for larger systems, our witnessing technique will outperform any strategy that demands full-state tomography.

The stabilizer nullity can be extended to treat unitaries~\cite{jiang2023lower}. In this form, it gives a lower bound to the number of $T$ gates required to apply a certain unitary. A similar property holds for other magic monotones. Bounding the number of $T$ gates is not possible with our formalism, as this is a profoundly basis-dependent characterization.

It is noteworthy that several of the quantifiers mentioned above have been associated with the perspective of witnessing as can be seen in Ref.~\cite{seddon2021quantifying}. Therein some witnesses for single-qubit magic states were proposed. Additionally, the authors also show curious lower and upper bounds on the scaling of magic monotones as the number of qubits $n$ grows. 

\subsection{Methods that are not valid for generic mixed states}

To the best of our knowledge, various quantifiers have not been generalized beyond pure states; examples of these include stabilizer fidelity, stabilizer rank, and stabilizer nullity. On the other hand, the stabilizer extent has a mixed state version~\cite{seddon2021quantifying}.

Recall that if one assumes purity of states and multi-qubit systems, a trivial witness is to measure a single overlap of the state with respect to the $\vert 0^n\rangle$ state, in which case deviations from $1/2^n$ will witness nonstabilizerness. Because of that, it is only quantification that proves to be a non-trivial task in the case of multi-qubit pure states. Therefore, in this section, we focus on quantification methods that are in some way efficient to calculate, or measure, at the cost of being defined only for pure states (or specific classes of mixed states). The stabilizer entropies are the most relevant monotones in this category. The first such entropy introduced was the stabilizer Rényi entropy~\cite{leone2022stabilizer}. While this quantity cannot be efficiently measured in general, it can be efficiently computed. Ref.~\cite{haug2023stabilizerentropies} investigated when such functions can be considered monotones and provided an explicit case study where stabilizer entropies could be computed efficiently. Ref.~\cite{haug2023efficient} introduced novel stabilizer entropies that can also be efficiently measured, but are still defined only for pure states.

It is simple to see for the case of stabilizer Rényi entropies why the quantifiers do not hold for generic mixed states. Let us consider, for instance, the stabilizer $2$-Rényi entropy defined for mixed states $\rho \in \mathcal{D}(\mathbb{C}^2)$ as
\begin{equation}
    \tilde{M}_2(\rho) \coloneqq -\log_2 \left(\frac{\sum_{P}(\text{Tr}(P\rho))^4}{2 \text{Tr} (\rho^2)}\right)
\end{equation}
where and the sum is taken over the $+1$ elements of the single-qubit Pauli group $\{\mathbb{1}, X, Y, Z\}$. Crucially, $\Tilde{M}$ is \emph{not} a monotone to any mixed state $\rho$, but only those with a specific form, given by $\rho = \frac{\mathbb{1}}{2}+ \frac{1}{2}\sum_{P \in G}\phi_P P$ where $G$ is a subset of the single-qubit Pauli group, and $\phi_P \in \{-1,1\}$~\cite{leone2022stabilizer}. For instance, considering the magic state
\begin{equation*}
    \vert F \rangle \langle F \vert = \frac{1}{2}\left(\mathbb{1}+\frac{1}{\sqrt{3}}(X+Y+Z)\right)
\end{equation*}
and mixed states $\rho_\nu = \mathcal{E}_\nu(\vert F \rangle \langle F \vert ) = (1-\nu) \vert F\rangle \langle F \vert + \nu\frac{\mathbb{1}}{2},$
the monotone as a function of $\nu$ becomes
\begin{equation}
    \Tilde{M}_2(\rho_\nu) = -\log_2 \Bigg\{\frac{1+3\left[(1-\nu)/\sqrt{3}\right]^4}{1+3\left[(1-\nu)/\sqrt{3}\right]^2}\Bigg\}.
\end{equation}
Since $\Tilde{M}_2(\rho_\nu)>0$ for any $\nu>0$, it is not a faithful monotone (or witness), since it would signal the presence of nonstabilizerness for states arbitrarily close to the maximally mixed state. 

Despite not being defined for mixed states, these monotones are extremely promising ways of efficiently estimating nonstabilizerness. For instance, Ref.~\cite{haug2023scalable, oliviero2022measuring} estimated magic in a cloud-available quantum computer. Ref.~\cite{haug2023efficient} showed that there are stabilizer entropies that can be \emph{efficiently measured}, with the required number of measurements (or post-processing) being independent of system size. Finally, Ref.~\cite{haug2023scalable} introduced a novel monotone, which the authors termed ``Bell magic'', that besides efficiently estimated on a quantum computer via Bell measurements also generalizes to certain sets of mixed states. Bell magic was recently measured in Ref.~\cite{bluvstein2023logical}.

\subsection{Methods that are semi-device independent}

Due to the connection between magic, the negativity of quasiprobability distributions and noncontextuality (both Kochen--Specker and generalized), under certain considerations, any test of such notions of classicality will also be a test capable of witnessing nonstabilizerness. This was already discussed in Appendix~\ref{app: contextuality}. Here, we comment on the device- or semi-device-independence of these tests.

Our witnesses are inequality-based and semi-device independent in the sense that (i) the test is made based only on statistics arising from two-state overlaps and (ii) we assume (in most cases) that the underlying system is a multi-qubit system. These restrictions are the ``semi'' for our approach. Inequalities that can be used to witness noncontextuality are significantly more device-independent, in the sense that they are not necessarily overlap-based (or correlation-based), while they will necessarily require some prior information about the Hilbert space considered: even, odd,  odd-prime, or composite system dimensionality structures.

One subtle point needs to be made:  Violations of noncontextuality inequalities \emph{per se} do not suffice to experimentally witness the failure of noncontextual explanations of the data. One must also test that the experimental requirements relative to the notions of KS-noncontextuality or generalized noncontextuality are operationally met. Similarly to Bell inequality violations, merely violating them does not attest to the failure of a local explanation of the data; some minimal requirements need to be met (such as space-like separation between parties, no-signaling, etc.). In our case, the requirement is that the data is described by two-state overlaps, while in noncontextuality inequalities other requirements are needed, and should be taken into consideration, even if one is only interested in witnessing nonstabilizerness due to contextuality.

\section{Proof of Theorem~\ref{theorem: h4 is a magic witness}}\label{app: hm family}

In this Appendix, we start by building a series of results that are used to facilitate the proof of Theorem~\ref{theorem: h4 is a magic witness}.

The inequality $h_4(r)\leq 1$ is facet-defining for $\mathfrak{C}(\mathcal{K}_4)$, where $\mathcal{K}_4$ is the complete graph of four vertices. This inequality is given by
\begin{equation}\label{eq: h4}
    h_4(r) = r_{1,2}+r_{1,3}+r_{1,4}-r_{2,3}-r_{2,4}-r_{3,4} \leq 1.
\end{equation}

We now demonstrate the following lemma.
\begin{lemma}\label{lemma: K4 equal labels}
    Let $\mathcal{G} = \mathcal{K}_m$. If $r=r(\underline\psi)$ such that the $\underline\psi$-labeling $\ell: V(\mathcal{G}) \to \underline\psi$ assigns the same pure state to adjacent vertices (any pair of vertices sharing an edge), then $h_m(r(\underline\psi)) \leq 1$.
    \label{lemma: no state repetition}
\end{lemma}

\begin{proof}
    Without loss of generality, we may consider $r_{1,k^\star}=1$ for some $k^{\star} \neq 1$. Let $r=r(\underline\psi)$ be any pure state realization satisfying this constraint. In this case, we must have that $\vert \psi_1\rangle = \vert \psi_{k^\star}\rangle$. Therefore, $r_{1,k}=r_{k^\star, k}$ for all $k=\{2,\dots,m\}\backslash\{k^{\star}\}$. The inequality $h_m(r) \leq 1$ is then written as 
    \begin{align*}
        h_m(r) &= \sum_{k=2}^m r_{1,k}-\sum_{i=2}^{m-1} \sum_{j>i}^mr_{i,j} \\&= 1+\sum_{{\begin{array}{c}k=2\\k\neq k^\star\end{array}}}^m r_{k^\star, k}-\sum_{i=2}^{m-1} \sum_{j>i}^mr_{i,j}\\
        &= 1-\sum_{{\begin{array}{c}i=2\\i\neq k^\star\end{array}}}^{m-1} \left( \sum_{{\begin{array}{c}j>i\\j\neq k^\star\end{array}}}^m r_{i,j} \right) \leq 1
    \end{align*}
    where we have used the fact that every element $r_{k^\star, k}$ is present in the sum $\sum_{i=2}^{m-1} \sum_{j>i}^mr_{i,j}$.
\end{proof}

Recall that $\mathfrak{Q}(\mathcal{G})$ is the set of all quantum realizable edge-weights given an event graph $\mathcal{G}$. Any stabilizer realization is (evidently) a quantum realization. In the remainder of this Appendix, we will focus on situations concerning stabilizer realizations.

Let us now show that, assuming an edge-weight is non-zero then its stabilizer realization has a lower bound.
\begin{lemma}\label{lemma: lower bound STAB overlaps}
    Take three (arbitrary) $n$-qubit stabilizer states $\ket{\psi_1}$, $\vert \psi_2 \rangle$, and $\vert \psi_3 \rangle$ such that $r_{1,2} = 1/2^{N_2}$, $r_{1,3} = 1/2^{N_3}$ and $r_{2,3}\neq 0$. Then, $r_{2,3} \geq 1/2^{N_2+N_3}$ where $N_2,N_3\in \{0,\dots,n\}$.
\end{lemma}

\begin{proof}
    Without loss of generality, we can take $\ket{\psi_1} = \ket{0^n}$. Let us start by considering the case where both $N_2,N_3 \neq 0$ or, equivalently, where $r_{1,2},r_{1,3} \neq 1$. Since $r_{1,2} = 1/2^{N_2}$, we have that $\vert \psi_2 \rangle$ is of the form 
    \begin{equation*}
        \vert \psi_2 \rangle = \frac{1}{2^{N_2/2}} \left( \vert 0^n\rangle + \sum_{j=1}^{2^{N_2}-1} i^{\alpha_j} \ket{a_j}  \right)
    \end{equation*}
    where $\alpha_j \in \mathbb{Z}_4$ and $a_j \in \mathbb{F}_2^n\backslash\{0^n\}$.
    Similarly, $r_{1,3} = 1/2^{N_3}$ implies that the state $\vert \psi_3 \rangle$ is of the form 
    \begin{equation*}
        \vert \psi_3 \rangle = \frac{1}{2^{N_3/2}} \left( \vert 0^n\rangle + \sum_{j=1}^{2^{N_3}-1} i^{\beta_j} \ket{b_j}  \right),
    \end{equation*}
    where $\beta_j \in \mathbb{Z}_4$ and $b_j \in \mathbb{F}_2^n\backslash\{0^n\}$. From this, we see that
    \begin{equation*}
        r_{2,3} = \frac{1}{2^{N_2+N_3}} \left| 1 + \sum_{j,j^{\prime}} i^{\alpha_j - \beta_{j^{\prime}}} \braket{b_{j^{\prime}}|a_j} \right|^2.
    \end{equation*}
   If $r_{2,3} \neq 0$, it is clear from the expression above that $r_{2,3}\geq 1/(2^{N_2+N_3})$.

   Finally, we note that, if either $N_i= 0$, we have the corresponding state $\ket{\psi_i} = \ket{0^n}$. In that case, it is clear that $r_{2,3} = 1/2^{N_j}$, with $j\neq i$, which complies with the lower bound established above. 
\end{proof}

Next, we demonstrate a result concerning realizations containing orthogonal states. This is the most important stepping stone to the proof of Theorem~\ref{theorem: h4 is a magic witness} because realizations involving null edge-weights are significantly harder to analyze with respect to inequality violations. 
\begin{lemma}
    Let $\mathcal{G} = \mathcal{K}_4$ and consider a quantum realization $r=r(\underline\rho_{\mathrm{STAB}}) \in \mathfrak{Q}(\mathcal{K}_4)$, where $\mathrm{STAB}$ denotes the set of $n$-qubit stabilizer states. If such a realization assigns to any two vertices two orthogonal states, then $h_4(r) \leq 1$.
    \label{lemma: no orthogonal states}
\end{lemma}
\begin{proof}
    Consider the set of four (pure) $n$-qubit stabilizer states: $\{ \ket{\psi_1},\,\ket{\psi_2},\,\ket{\psi_3},\,\ket{\psi_4} \}$, where $\ket{\psi_i}$ is the stabilizer state associated with the $i$th vertex.
    The following observations follow trivially from Lemma~\ref{lemma: no state repetition} when considering realizations with stabilizer states: (i) If four or more overlaps are zero, Eq.~\eqref{eq: h4} cannot be violated; (ii) If $\ket{\psi_1}$ is orthogonal to any of the other states, again no violation of Eq.~\eqref{eq: h4} is possible; (iii) To achieve a violation, at least two of the overlaps $r_{1,j}$ must equal $1/2$ and the remaining overlap with positive contribution must obey $r_{1,k}>r_{2,3} + r_{2,4} + r_{3,4}$ with $k\neq j$. 

    Throughout, we take $r_{1,2} = 1/2$ and $\ket{\psi_2} = \ket{0^n}$ which imposes that $\ket{\psi_1} = (\ket{0^n} + i^{\alpha}\ket{s})/\sqrt{2},$ where $\alpha \in \mathbb{Z}_4$, $s\in \mathbb{F}_2^n\backslash\{0^n\}$, and $\ket{s}$ denotes the corresponding computational-basis state. Moreover, we can set $\alpha = 0$ because there is always a Clifford unitary that transforms $(\ket{0^n} + i^{\alpha}\ket{s})/\sqrt{2}$ into $(\ket{0^n} + \ket{s})/\sqrt{2}$ while leaving the state $\ket{0^n}$ unchanged. Thus, for simplicity, we take $\ket{\psi_1} = (\ket{0^n} + \ket{s})/\sqrt{2}$. All of this is done wlog.
    
    If we have three overlaps equal to zero, the only way for a violation to occur is that: (i)~$r_{2,3} = r_{2,4} = r_{3,4} = 0,$ (ii)~$r_{1,3} = 1/2$, and (iii)~$r_{1,4}>0$. Note that the roles of $r_{1,3}$ and $r_{1,4}$ could be switched, leading exactly to the same conclusion. We will now show that these three conditions are incompatible. A generic stabilizer state takes the form:
    \begin{equation}\label{eq: generic stabilizer state}
        \ket{\psi_j} = \frac{1}{\sqrt{|\mathcal{W}_j|}}\sum_{w\in \mathcal{W}_j} i^{\alpha_w} \ket{w}\,
    \end{equation}
    where $\mathcal{W}_j \subseteq \mathbb{F}_2^n$ and $|\mathcal{W}_j| = 2^{N_j}$, for some $N_j\in\{0,n\}$. For the state to be a stabilizer state, $\mathcal{W}_j$ and $\alpha_w$ must possess specific properties; these are irrelevant for the purposes of our proof and we will therefore omit them, but the interested reader is pointed to Appendix~A of Ref.~\cite{Nest2010beyondGK} or Theorem~9 of Ref.~\cite{garcia2017onthegeometry} for details.
    
    Since, $r_{2,3} = 0$, for the state $\ket{\psi_3}$, $0^n$ cannot belong to $\mathcal{W}_3$. Combining this observation with condition (ii), and given that $\ket{\psi_1} = (\ket{0^n} + \ket{s})/\sqrt{2},$ it clear that $\ket{\psi_3} = \ket{s}$. Because $r_{2,4} = r_{3,4} = 0$, the remaining state $\ket{\psi_4}$ must be a linear combination of computational-basis states so that $0^n,s \notin \mathcal{W}_4$. This necessarily means that $r_{1,4} = 0$, violating condition~(iii). 
    
    If exactly two of the overlaps $\{r_{2,3}, r_{2,4},r_{3,4}\}$ are zero, Fig.~\ref{fig: ort_proof}(a) illustrates the nine possible combinations of edge-weight assignments that could potentially lead to violations. Fortunately, symmetry constraints illustrated therein mean that we can restrict ourselves to only two different sub-cases.
    \begin{figure}[t]
    \centering
    \includegraphics[width=\columnwidth]{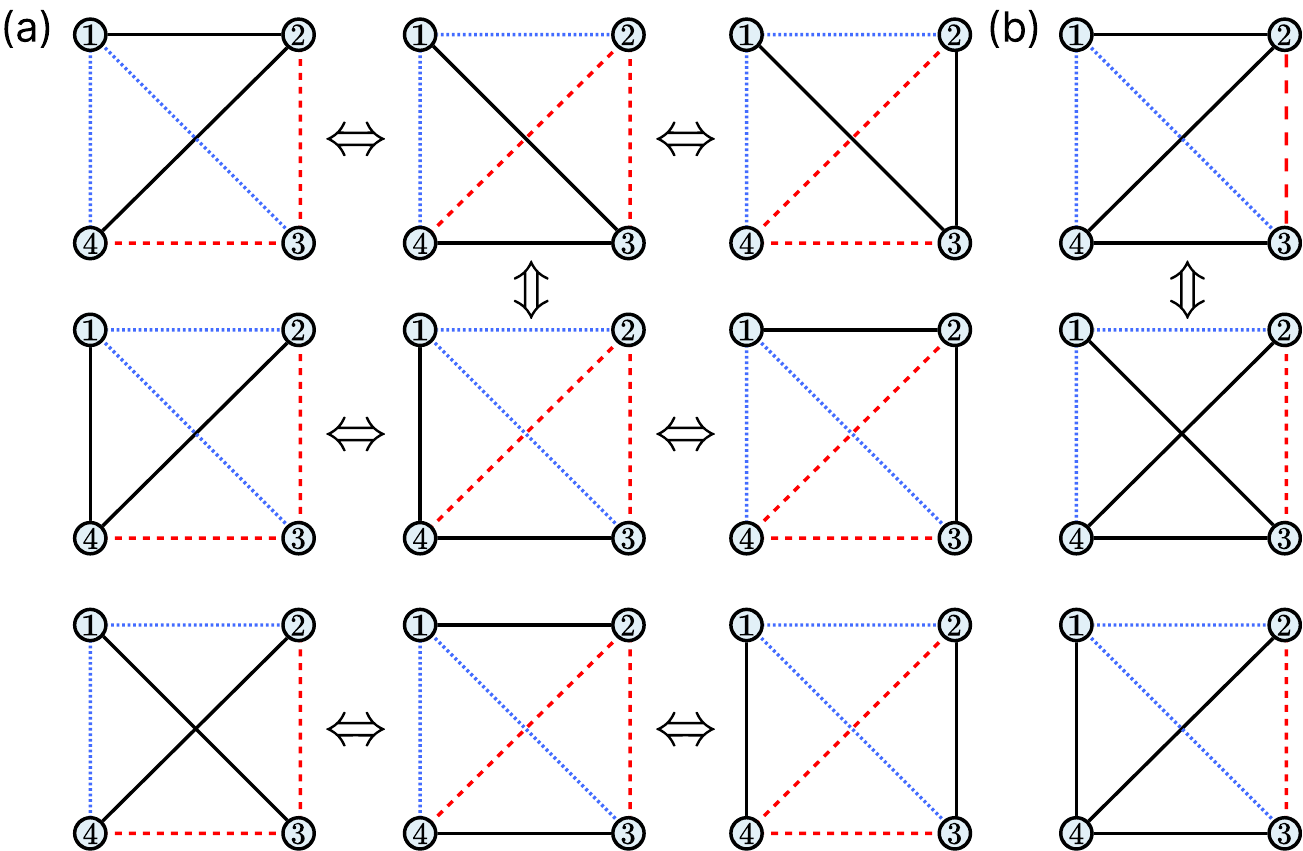}
    \caption{\textbf{Vertex assignments with null edge-weights.} The figure illustrates the cases where (a) exactly two overlaps are zero and (b) there is exactly one null overlap. Dashed-red lines indicate edges of the graph with an assigned weight equal to zero, dotted-blue lines represent the edges assigned with the value 1/2, while solid-black lines depict edges with positive (but arbitrary) edge-weight.}
    \label{fig: ort_proof}
    \end{figure}

    \textit{Subcase1.--}We consider the case depicted in the second row of the second column of Fig.~\ref{fig: ort_proof}(a): $r_{1,2} = r_{1,3} = 1/2$ and $r_{2,3} = r_{2,4} = 0.$ Every scenario in the top two rows is equivalent to this one.
    
    For a violation to occur $r_{1,4} > r_{3,4}.$ The fact that $r_{1,3} = 1/2$ and $r_{2,3} = 0$ implies that $\ket{\psi_3} = \ket{s}.$ The state placed at the fourth vertex can admit the form in Eq.~\eqref{eq: generic stabilizer state}. Because $r_{2,4}=0$ then $0^n\notin \mathcal{W}_4$; contrarily, since $r_{3,4}\neq 0$, $s\in \mathcal{W}_4.$ Immediately, this means that $r_{1,4} = 1/2^{N_4+1}$ while $r_{3,4} = 1/2^{N_4},$ which means that $r_{3,4} > r_{1,4},$ and therefore no violation can occur. 

    \textit{Subcase2.--} We now consider the cases in the bottom row of Fig.~\ref{fig: ort_proof}(a) which are all equivalent between themselves, but non-equivalent to the six cases in the top two rows. We take the edge-weight assignment on the first column: $r_{1,2} = r_{1,4} = 1/2$ and $r_{2,3} = r_{3,4} = 0.$ This means that, for a violation to occur $r_{1,3} > r_{2,4}.$ The condition that $r_{1,4} = 1/2$ enforces one of the following three forms for $\ket{\psi_4}:$
    \begin{equation*}
        \ket{\psi_4} = 
        \begin{cases}
            \ket{0^n} \\
            \ket{s} \\
            \frac{\ket{0^n} + i^{\beta}\ket{s}}{\sqrt{2}}, \beta\in\{1,3\}\,. 
        \end{cases}
    \end{equation*}
    Lemma~\ref{lemma: no state repetition} informs us that the first option will lead to no violation; moreover, the second option leads to $r_{2,4} = 0$ taking us back to the three-null-overlaps situation which we already proved leads to no violation. This leaves us with the last option, that is: $\ket{\psi_4} = (\ket{0^n} + i^{\beta}\ket{s})\sqrt{2},$ with $\beta\in \{1,3\}$. Immediately, we see that $r_{2,4} = 1/2$, and therefore it is impossible to meet that condition that $r_{1,3} > r_{2,4},$ so that no violation is possible in this case either.

    Finally, it remains to assess the case where only a single null overlap exists. Fig.~\ref{fig: ort_proof}(b) illustrates the three possible edge-weight assignments that may lead to violations. We note that the top two cases are equivalent, meaning that, again, we have to focus only on two subcases.

    \textit{Subcase1.--}We consider $r_{1,2} = r_{1,4} = 1/2$ and $r_{2,3} = 0$ (second row of Fig.~\ref{fig: ort_proof}(b)). For a violation to occur, the following must hold $r_{1,3} > r_{2,4} + r_{3,4}.$ We note that $r_{1,4} = 1/2$ and $r_{2,4} \neq 0$ implies that $\ket{\psi_4} = (\ket{0} + i^{\beta}\ket{s})/\sqrt{2}$, with $\beta= \{1,3\}$. This fixes $r_{2,4} = 1/2$ which immediately means that the condition $r_{1,3} > r_{2,4} + r_{3,4}$ can never be met, and therefore no violation can occur.
    
    \textit{Subcase2.--}We consider $r_{1,2} = r_{1,3} = 1/2$ and $r_{2,3} = 0.$ This means that for a violation to hold, we must have $r_{1,4} > r_{2,4} + r_{3,4}.$ The fact that $r_{1,3} = 1/2$ and $r_{2,3} = 0$ implies that $\ket{\psi_3} = \ket{s}.$ The state in the fourth vertex can assume the general form given by Eq.~\eqref{eq: generic stabilizer state} where both $0^n$ and $s$ must belong to the set $\mathcal{W}_4$ (otherwise, we fall back into the cases with two or three null overlaps). Automatically this means that $r_{2,4} = r_{3,4} = 1/2^{N_4}.$ On the other hand, $r_{1,4}$ can be (at most) $r_{1,4} = 1/2^{N_4-1},$ which means that the condition $r_{1,4} > r_{2,4} + r_{3,4}$ cannot be met, and therefore no violation can occur.

    This concludes the assessment of all possible cases. Therefore, if any overlap $r_{i,j}$ is zero, no violation of the inequality~\eqref{eq: h4} is possible.
\end{proof}

Finally, we have all the tools needed to prove Theorem~\ref{theorem: h4 is a magic witness} of the main text.

\begin{proof}[Proof of Theorem~\ref{theorem: h4 is a magic witness}]
    Consider a set of four (pure) $n$-qubit stabilizer states: $\{ \ket{\psi_1},\,\ket{\psi_2},\,\ket{\psi_3},\,\ket{\psi_4} \}$. Recall that, if we want to find a violation of Eq.~\eqref{eq: h4}, no two states can be the same (Lemma~\ref{lemma: no state repetition}) so that: $\ket{\psi_i} \neq \ket{\psi_j}$ for $i\neq j$. Therefore, to obtain $h_4 > 1$ the following conditions must hold: (i) there are at least two $r_{1,j} = 1/2\,$, (ii) the remaining overlap $r_{1,k}$ with $k\neq j$ must obey: $r_{1,k} > r_{2,3} + r_{2,4} + r_{3,4}\,.$

    Without loss of generality, take $\ket{\psi_1} = \ket{0^n}$ and assume that $r_{1,2} = r_{1,4} = 1/2\,.$ Under these assumptions, for a violation to occur we must have $r_{1,3} > r_{2,3} + r_{2,4} + r_{3,4}\,.$
    
    Because $r_{1,2} = 1/2$ this means that $\ket{\psi_2} = (\ket{0^n} + \ket{s})/\sqrt{2}$ where $s\in \mathbb{F}_2^n\backslash\{0^n\}$. Evidently, something similar can be said for $\ket{\psi_4}:$ $\ket{\psi_4} = (\ket{0^n} + i^{\alpha} \ket{w})/\sqrt{2}$, with $w\in \mathbb{F}_2^n\backslash\{0^n\}$ and $\alpha\in \mathbb{Z}_4$. This will impose a constraint on the overlap $r_{2,4}:$
    \begin{equation}
        r_{2,4} = 
        \begin{cases}
        1,\quad\quad \text{ if } w=s \land \alpha = 0 \\
        0,\quad\quad \text{ if } w=s \land \alpha = 2 \\
        1/2,\quad \text{ if } w=s \land \alpha = \{1,3\}\\
        1/4,\quad \text{ if } w\neq s
        \end{cases}\,.
        \label{eq: overlap options}
    \end{equation}
    Lemmas~\ref{lemma: no state repetition} and~\ref{lemma: no orthogonal states} guarantee, respectively, that the first and second options give no violation and we can thus focus on the other two.

    Taking $r_{2,4} = 1/2$ we have: $r_{1,3} > 1/2 + r_{2,3} + r_{3,4},$ which is impossible to verify since at most $r_{1,3}$ can be $1/2$.

    Taking $r_{2,4} = 1/4$ we get the condition $r_{1,3} > 1/4 + r_{2,3} + r_{3,4}.$ In order for this to hold, $r_{1,3}$ must be equal to $1/2$ which implies that the corresponding state must take the form $\ket{\psi_3} = (\ket{0^n} + i^{\beta}\ket{t})/\sqrt{2}$. Immediately, this will imply that the values for $r_{2,3}$ and $r_{3,4}$ are bounded as $r_{2,4}$ in Eq.~\eqref{eq: overlap options}. Since we know realizations with null overlaps to yield no violation (Lemma~\ref{lemma: no orthogonal states}), it is clear that the condition $r_{1,3} > 1/4 + r_{2,3} + r_{3,4}$ can never be met, because $r_{1,3}$ can be at most 1/2.

    Theorem~\ref{theorem: functional max} guarantees that this holds also for mixed stabilizer states. This concludes the proof.
\end{proof}

We conclude this Appendix by generalizing Thereom~\ref{theorem: h4 is a magic witness} for $d$-dimensional qudits.

\begin{theorem}
    There exists no quantum realization $r=r(\underline\rho_{\mathrm{STAB}(d)}) \in \mathfrak{Q}(\mathcal{K}_4)$, where $\mathrm{STAB}(d)$ denotes the set of stabilizer states of $n$ $d$-dimensional qudits, such that $h_4(r) >1$, for any integer value $n\geq 1$.
    \label{theorem: no violation from qudit stabilizer states}
\end{theorem}
\begin{proof}
    The overlap $\left| \langle \psi \vert \phi \rangle \right|^2$ of any two non-orthogonal stabilizer states, $\ket{\psi}$ and $\ket{\phi},$ of $n$ qudits of dimension $d$ can assume value $1/d^{N}$ with $0 \leq N \leq n\,$, see Lemma~2 of Ref.~\cite{kueng2015qubit} on the overlap between pure stabilizer states. 

    Again, Lemma~\ref{lemma: no state repetition} guarantees that, if any state is repeated, there is no violation of the inequality, so we assume all states to be different. The positive terms in $h_4(r)$ lead to $r_{1,2} + r_{1,3} + r_{1,4} = 3/d\,,$ in the best case. For $d = 3$ this leads to a value of at most $1$ and for $d>3$ the value will be smaller than one. As a consequence, it is immediately realized that no violation of the inequality is possible.
\end{proof}

\section{Full set magic and overlap estimation}\label{app: full set magic}

In this Appendix, we present numerical evidence for the ability of odd cycle inequalities to witness full set magic. As can bee seen from Table~\ref{tab: Number of magic states}, letting $\underline\rho = \{\rho_i\}_{i=1}^m \subseteq \mathcal{D}(\mathbb{C}^2)$ be the set of states such that $c_m(r(\underline\rho))>m-2$, we have that at least $m-2$ states of $\underline\rho$ must be magic states. Depending on how large the violation is, we have that, other than one reference state that can always be unitarily sent inside STAB, all other states must be magic. The second column from Table~\ref{tab: Number of magic states} is obtained as follows. We assume that there are at least two states in a set $\{\vert \psi_i \rangle\}_{i=1}^m$ that are single-qubit stabilizer states. Without loss of generality, we can take $\{\vert 0\rangle, \vert +\rangle\}$. All the remaining $m-2$ states are generic single-qubit states of the form $\ket{\psi_i} = \cos(\theta_i)\vert 0\rangle + e^{i \phi_i}\sin(\theta_i)\vert 1 \rangle$. For all possible combinations of $m-2$ generic states with the two chosen stabilizer states (or, equivalently, all possible edge $\underline\psi$-labelings of the graph  $\mathcal{C}_m$ permuting the two STAB states), we maximize $c_m$ with respect to the variables $\underline\theta = \{\theta_i\}_i$ and $\underline\phi = \{\phi_i\}_i$.

Let us take $c_4$ as an example. The set of quantum states is $\{\vert 0\rangle, \vert +\rangle, \vert \psi_1\rangle, \vert \psi_2\rangle\}$ and two possible maximizations are $c_m(r(\ell_1(\underline\psi))) = -r_{0,+}+r_{+,1}+r_{1,2}+r_{2,0}$
and $c_m(r(\ell_2(\underline\psi))) = -r_{0,1}+r_{1,+}+r_{+,2}+r_{2,0}.$
The tool used for the maximization was \textsf{NMaximize} of Wolfram Mathematica and all available solvers were tested, with the final result presented as the one that found the largest value. Note that, as before, even though we are maximizing for pure states, the results presented in Table~\ref{tab: Number of magic states} hold for generic mixed states. The last column corresponds to the situation wherein the set is allowed to be full set magic, i.e., all states but one are magic.

Interestingly, for even cycle inequalities, we cannot witness full set magic since $m-2$ magic states are sufficient to maximally violate the inequality. Moreover, we can see that the gap between set magic and full set magic decreases as $m$ increases.
\begin{table}[b]
    \centering
    \begin{tabular}{ccc}
    \hline\hline
       $m$  &  $c_m$  &  $c_m^{max}$ \\
    \hline
        3  & 1.2071 & 1.2500 \\
        4  & 2.4142 &  2.4142 \\
        5  & 3.5061 & 3.5225 \\
        6  & 4.5981 & 4.5981 \\
        7  & 5.6468 & 5.6534 \\
        8  & 6.6955 & 6.6955 \\
        9  & 7.7254 & 7.7286 \\
    \hline\hline
    \end{tabular}
    \caption{\textbf{Full set magic bounds.} The first column shows the cycle inequality considered and the second column shows that optimal value considering quantum realizations of sets of states $\underline\rho \subseteq \mathcal{D}(\mathbb{C}^2)$ where at least two states $\rho_i,\rho_j \in \mathrm{STAB}$. The last column presents the optimal tight values found in Appendix~\ref{app: SDP}. Interestingly, even cycles cannot witness full set magic, but odd cycles can.}
    \label{tab: Number of magic states}
\end{table}

We take this opportunity to discuss different strategies for evaluating two-state overlaps. One option is to consider prepare-and-measure estimation. In this case, a preparation stage prepares states $\vert \psi_i\rangle = U_i \vert 0\rangle$, while a measurement stage acts as a projection onto $\langle \psi_j \vert = \langle 0 \vert  U_j^\dagger$. Assuming we are in the regime of full set magic, each overlap inequality will certify the nonstabilizerness of all the preparations $U_i$ \emph{and} all measurements $U_j$.

Another possibility is to use the SWAP-test~\cite{buhrman2001quantum} which, despite using a nonstabilizer gate (the Fredkin gate), can unambiguously witness nonstabilizerness of the quantum states that are sent by a third party (sender). If we want to avoid the use of auxiliary qubits and nonstabilizer operations to estimate the overlaps, we can instead use Bell measurements~\cite{haug2023scalable}.

Finally, in the context of linear-optical implementations, we can use the Hong-Ou-Mandel effect~\cite{hom87}. Any linear-optical interferometer that is insensitive to internal degrees of freedom of the photons (e.g. polarization, frequency, time of arrival) has outcomes determined only by unitary-invariant properties of the spectral functions describing them \cite{shchesnovich2018collective, giordani2021witnesses}. This means the Hong-Ou-Mandel effect, or interferometry in more complex multimode interferometers, can be used to directly estimate those invariant properties. Such a test could be used to certify the nonstabilizerness of these states.

\end{document}